\documentclass[%
 reprint,
superscriptaddress,
nofootinbib,
longbibliography,
noeprint,
 amsmath,amssymb,
 aps,
prb,
]{revtex4-2}

\usepackage{graphicx}
\usepackage{dcolumn}
\usepackage{bm}
\usepackage{amsthm} 
\usepackage[ruled,vlined, linesnumbered]{algorithm2e}
\usepackage{hyperref}
\usepackage[caption=false,justification=centerlast]{subfig}
\newtheorem{prop}{Proposition}
\newtheorem{theorem}{Theorem}
\theoremstyle{definition}
\newtheorem{definition}{Definition}
\usepackage [autostyle, english = american]{csquotes}
\MakeOuterQuote{"}
\usepackage{color}
\newcommand{\gs}[1]{\textcolor{blue}{[GS:#1]}}
\newcommand{\comment}[1]{}
\renewcommand{\vec}{\mathbf}
\usepackage{url}
\begin{document}
\date{\today} 

\title{From hard spheres to hard-core spins}

\author{Grace M. Sommers}
\affiliation{Department of Physics, Princeton University, Princeton, NJ 08540, USA}

\author{Benedikt Placke}
\affiliation{Max-Planck-Institut f\"{u}r Physik komplexer Systeme, 01187 Dresden, Germany}

\author{Roderich Moessner}
\affiliation{Max-Planck-Institut f\"{u}r Physik komplexer Systeme, 01187 Dresden, Germany}

\author{S. L. Sondhi}
\affiliation{Department of Physics, Princeton University, Princeton, NJ 08540, USA}

\begin{abstract}
A system of hard spheres exhibits physics that is controlled only by their density. This comes about because the interaction energy is either infinite or zero, so all allowed configurations have exactly the same energy. The low density phase is liquid, while the high density phase is crystalline, an example of "order by disorder" as it is driven purely by entropic considerations. Here we study a family of hard spin models, which we call hardcore spin models, where we replace the translational degrees of freedom of hard spheres with the orientational degrees of freedom of lattice spins. Their hardcore interaction serves analogously to divide configurations of the many spin system into allowed and disallowed sectors. We present detailed results on the square lattice in $d=2$ for a set of models with $\mathbb{Z}_n$ symmetry, which generalize Potts models, and their $U(1)$ limits, for ferromagnetic and antiferromagnetic senses of the interaction, which we refer to as exclusion and inclusion models. As the exclusion/inclusion angles are varied, we find a Kosterlitz-Thouless phase transition between a disordered phase and an ordered phase with quasi-long-ranged order, which is the form order by disorder takes in these systems. These results follow from a set of height representations, an ergodic cluster algorithm, and transfer matrix calculations.
\end{abstract}

\maketitle
    
\section{Introduction}

Systems with constraints but no other interactions are a fascinating corner of statistical mechanics on three counts. First, their equilbrium physics is purely entropic, so any ordering they exhibit is "order by disorder." The canonical example of this somewhat counterintuitive phenomenon is nematic ordering in Onsager's model of thin, hard rods~\cite{Onsager1949}. While at low densities this system lacks orientational order, with rods arranged isotropically, at a critical density the increase in \textit{translational} entropy afforded by aligning the rods outweighs the decrease in \textit{orientational} entropy that such an alignment evidently entails. Second, such systems probe universality in a non-trivial fashion. Viewed as interacting systems, they involve interactions of infinite strength, so many of the standard, perturbative arguments for long-wavelength universality do not directly apply to them. Finally, their dynamics has some simplifying features that have been used in simple cases to reach conclusions that are otherwise difficult. The most famous case of this is a set of results on ergodicity following the seminal work of Sinai~\cite{Sinai1963}.

In this paper we study a family of classical spin models in this class of "constraint-only" systems. Our models, most generally, involve of $M$ component spins $\vec{S}$ of fixed length $\vec{S}^2=1$ on a specified lattice. In this paper we do not specify their dynamics but only the pairwise additive energy function
\begin{equation}
    H = \sum_{\langle ij \rangle} V(\vec{S_i},\vec{S_j})
\end{equation}
where ${\langle ij \rangle}$ indicates that $i,j$ are nearest neighbor sites and the potential energy function has the form
\begin{equation}\label{eq:exclusion}
    V(\vec{S_i},\vec{S_j}) = \begin{cases} 0 &\mbox{if\ \ } |\vec{S_i} -\vec{S_j}| \geq \alpha \\ 
\infty & \mbox{otherwise.}  \end{cases} 
\end{equation}

As the parameter $\alpha$ moves between $0$ and $2$, the neighbors of a given spin are forced to lie outside a solid angle that increases from $0$ to eventually cover the entire unit sphere in $M$ dimensions on which the spins live. As the term hard spins is already reserved for spins of fixed length, we refer to these as hardcore spins. With the inequality as above we have an exclusion model of an antiferromagnetic persuasion, while with the sign of the inequality reversed we obtain an inclusion model. Explicitly, we define hardcore inclusion models by the pairwise potential:
\begin{equation}\label{eq:inclusion}
    V(\vec{S_i},\vec{S_j}) = \begin{cases} 0 &\mbox{if\ \ } |\vec{S_i} -\vec{S_j}| < \alpha \\ 
\infty & \mbox{otherwise}  \end{cases} 
\end{equation}
which is then of a ferromagnetic type. Finally, we will also consider discrete models in which only a finite number of points on the unit sphere are permitted.

One can view the exclusion models as representing a departure from Onsager's original problem in which the translational degrees of freedom are frozen while the orientational degrees of freedom, no longer directors but now spins, face a tunable set of local constraints. By contrast the well studied system of hard spheres/disks can be viewed as a departure in which the orientational freedom is removed and translational freedom retained. That system in three dimensions exhibits a first-order phase transition at a critical density above which crystallization occurs as a means to access greater free volume per particle and increase entropy~\cite{Frenkel2015}. The hard sphere system has been the subject of much work surrounding its dynamical properties and potential connections to the long standing problem of vitrification, and in a subsequent paper we will report results on the dynamics of the hardcore spins studied here. Indeed, the question of dynamics is what got us thinking about this family of models in the first place.

\begin{figure}[hbtp]
\subfloat[]{
\centering
\includegraphics[width=0.4\linewidth]{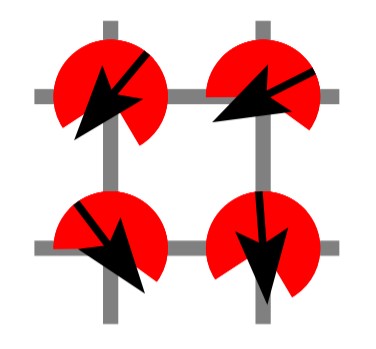}
\includegraphics[width=0.4\linewidth]{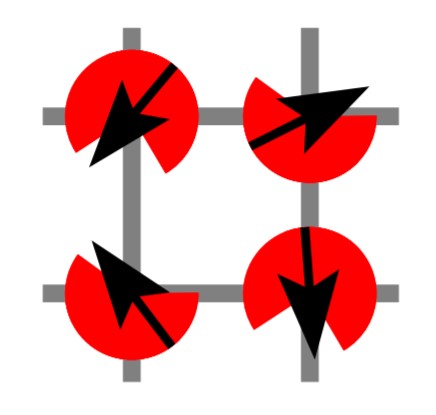}
\label{fig:example}
} \\
\subfloat[]{
    \centering
    \includegraphics[width=0.9\linewidth]{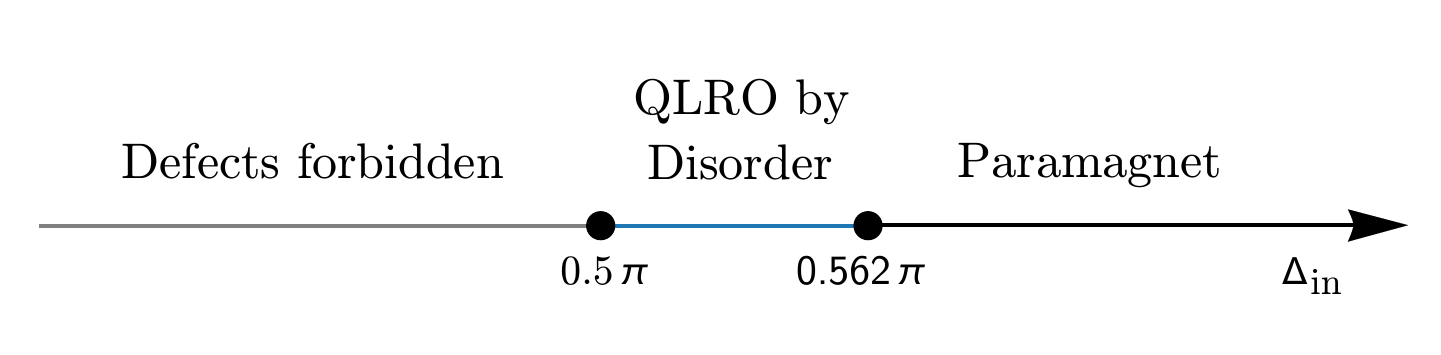}
    \label{fig:phaseline}
}
\caption{The new classes of hardcore spin models studied in this paper, in which the pairwise potential takes the form of Eq.~\ref{eq:exclusion} or Eq.~\ref{eq:inclusion}, specialized to the case $M=2$ on the square lattice. (a) Valid configurations for (left) an inclusion model with inclusion angle $\Delta=0.48\pi$ and (right) an exclusion model with exclusion angle $\pi-\Delta=0.52\pi$. Red shading indicates orientations forbidden by the nearest-neighbor interaction. (b) Sketch of the phase diagram in the XY limit of the inclusion model. Increasing the inclusion angle, defined in Eq.~\ref{eq:Vij}, tunes the system from a KT phase---consisting of both a vortex-forbidden region and a quasi-long-ranged order by disorder region---to a paramagnetic phase.}
\end{figure}

Returning to the statics, we focus in this paper on the simplest case of $M=2$ on the square lattice in $d=2$. When the spins are continuous, we get a hardcore spin model of XY spins (Fig.~\ref{fig:example}), and when they are discretized to take $N$ values, we obtain generalizations of Potts models with $\mathbb{Z}_N$ symmetry. In these cases we are able to search for order by disorder and for universal long wavelength physics dictated by the usual criteria of symmetry and dimensionality. For XY spins, true long-range order is, presumably, ruled out and only algebraically long-ranged order of the Kosterlitz-Thouless (KT) variety~\cite{Kosterlitz} is possible, although we are not aware of a proof that actually covers the hardcore limit. For the $Z_N$ cases, it would appear that true long-ranged order as well as quasi-long-ranged order are both possible on grounds of symmetry and as discussed in Ref.~\cite{Jose1977}.

Our principal results are (i) that both the $Z_N$ models for sufficiently large $N$ and the XY limiting model exhibit KT phases both when the constraint parameter forbids vortices but also for a range where they are allowed, thereby furnishing a case of order by disorder, (ii) that none of our $Z_N$ models or their XY limits exhibit a phase with true long-ranged order and (iii) that for $N \ge 4$ and extending into the XY limit, there are always models with nonzero interaction yet sufficiently small/large exclusion/inclusion angles which exhibit short-ranged correlations. We portray these results in the phase diagrams in Fig.~\ref{fig:phaseline} and Fig.~\ref{fig:phase}, of which the latter requires some definitions that we introduce in Section~\ref{sect:models} below. These results are obtained by a combination of a set of height representations for our $Z_N$ models, transfer matrix calculations, and  finite size scaling using a cluster algorithm, henceforth referred to as the \texttt{reflect} algorithm, whose ergodicity is proven for bipartite lattices.

The paper proceeds as follows. In Section~\ref{sect:models}, we define the models and review the theoretical context for analyzing them. We then describe, in Section~\ref{sect:methods}, the two numerical methods used to construct the phase diagram---the transfer matrix method and the cluster algorithm---as well as the height representations for the vortex free $Z_N$ models. In Section~\ref{sect:results}, we present results for $Z_N$ models and the XY limit, along with a discussion for how the phase diagram fits into and extends the general understanding of critical spin systems. We conclude in Section~\ref{sect:conclude}.

\section{Models and Background}\label{sect:models}
\subsection{Models and Notation}
Eq.~\ref{eq:exclusion} and its inclusion counterpart, Eq.~\ref{eq:inclusion}, express the hard constraint between neighboring spins in terms of the norm of the difference between two $M$-component unit vectors. Alternatively we could formulate this constraint in terms of the distance along the great circle connecting the two vectors, or the inner product between them. The latter formulation was used in a study of 
$O(M)$ and $RP^{M-1}$ constraint models for inclusion angles below $\pi/4$~\cite{Hasenbusch1996}. For $M=2$ these formulations are all equivalent. We find it most natural is to represent spins as complex exponentials $s_j = \exp(i\theta_j)$, where $\theta_j \in [0, 2\pi)$ is the orientation of the $j$th spin on the unit circle. In the XY limit, models are parameterized by an angle $\Delta$, with the nearest-neighbor interaction for inclusion models given by:
\begin{equation}\label{eq:XY-inclusion}
V(s_i, s_j) = 
\begin{cases}
0 & \mathrm{if \, angle}(s_i, s_j) < \Delta \\
\infty & \mathrm{otherwise}.
\end{cases}
\end{equation}
This can be viewed as the zero-temperature limit of a step model, which has mainly been studied at finite temperature for $\Delta=\pi/2$~\cite{Guttmann1973,Guttman1978,Lee1987,Irving1996} but was also examined, via the Migdal approximation, for smaller angles~\cite{Barber1983}. 

At zero temperature, a precision study of the phase transition for the "constraint action" described by Eq.~\ref{eq:XY-inclusion} found that it belongs to the Kosterlitz-Thouless universality class~\cite{Bietenholz2013, Bietenholz2013a}. Where our results overlap with these papers, we are in agreement. However, the present research goes beyond previous works in two major respects.

First, we generalize to $Z_N$ (clock) models, where spins can adopt $N$ orientations at angles $2\pi \sigma/N$ on the unit circle, with $\sigma\in\{0,1,...,N-1\}$~\cite{Cardy1980}. Using complex notation, we say that the spin at site $j$ is in the state $\sigma_j$ if $s_j = \exp(2\pi i \sigma_j/N)$. A hardcore inclusion model can then be characterized by a set of 3 parameters $(N, p_{in}, p')$ where neighboring spins must enclose $<p$ sites on a clock with $N$ orientations:
\begin{align}\label{eq:Vij}
V(s_i,s_j) = \begin{cases}
0 & \mathrm{if \, angle}(s_i, s_j) < 2\pi p / N \\
\infty & \mathrm{otherwise,}
\end{cases}
\end{align}
and $p'$ is the number of allowed orientations for spin $s_i$ given a fixed orientation of neighboring spin $s_j$.

A second generalization of the constraint action comes from inverting the piecewise relation in Eq.~\ref{eq:XY-inclusion} to create an exclusion model, wherein neighboring spins must enclose an angle $\geq\Delta$. Similarly, $Z_N$ exclusion models are parameterized by $(N, p_{ex}, p'$), where neighboring spins must enclose a minimum of $p$ sites on the clock, and $p'$ is again the number of allowed orientations for spin $s_i$ given a fixed orientation of spin $s_j$:
\begin{equation}
p'=N-2p_{ex}+1 = 2p_{in} - 1.
\end{equation}

The distinction between ferromagnetic inclusion and antiferromagnetic exclusion models on frustrated lattices gives rise to essential differences in the phase diagram~\cite{Obuchi2012,Lv2012}. In this paper, we instead focus on the square lattice, which for even widths and lengths is bipartite. This feature results in an equivalence between inclusion and exclusion models: if $S$ denotes a valid (zero-energy) inclusion configuration with inclusion angle $\Delta$, then flipping all the spins on one sublattice yields a zero-energy exclusion configuration $S'$ with angle $\pi - \Delta$, as depicted in Fig.~\ref{fig:example}. This equivalence is captured for finite $N$ by the parameter $p'$: an inclusion model with a given set of parameters $(N,p')$ is fully equivalent to an exclusion model with the same $(N, p')$. We can map both classes of models onto a common set of axes, $N$ vs. inclusion angle $\Delta$, by the relation:
\begin{equation}\label{eq:delta}
\frac{\pi}{N}(p'-1) < \Delta \leq \frac{\pi}{N}(p'+1).
\end{equation}
The inequality here comes from the fact that when the orientations of spins are discrete, there is ambiguity in how the real-valued inclusion angle is defined. In the following analysis, the most sensible definition of the "effective inclusion angle" will prove to be the midpoint of the interval, $\Delta=\pi p'/N$. In the XY limit ($N\rightarrow \infty$), the inclusion angle is uniquely defined.

The mapping between inclusion and exclusion means that we can primarily focus our analysis on inclusion models. However, the mapping only exists when the global rotation by $\pi$ on one sublattice is allowed by the discretization of the spins, namely, for even $N$. On the other hand, inclusion and exclusion models for odd $N$ probe a complementary set of parameters $p'$ (odd $p'$ for inclusion, even $p'$ for exclusion) as the clock lacks the $\mathbb{Z}_2$ symmetry of the lattice. In these cases, it will be necessary to separate examine exclusion models, to probe the entire parameter space. 

One striking example of an exclusion model with no equivalent inclusion model is the $(N,p_{ex}, p') = (3,1,2)$ exclusion model, which is simply a 3-state zero-temperature Potts antiferromagnet. Indeed, zero-temperature $q$-state Potts models can be viewed as special cases of our new classes of clock models, with $p=1$: in our notation, an $(N,1,1)$ inclusion model is a zero-temperature Potts ferromagnet with $q=N$, while an $(N,1,N-1)$ exclusion model is a zero-temperature Potts antiferromagnet (AFM).

Zero-temperature Potts AFMs, which possess macroscopic ground state entropy, have been the subject of extensive research over the past several decades, through evaluation of chromatic polynomials~\cite{Shrock1998,Chang2001,Jacobsen2003,Jacobsen2017}  along with formal proofs for the existence~\cite{Peled2018,Kotecky2014,Huang2013,Peled} or non-existence~\cite{Goldberg2005} of entropy-driven long-ranged order on various lattices. 
A given lattice has a critical value $q_c$, such that for $q>q_c$ a $q$-state Potts AFM is disordered at all temperatures, including $T=0$; for $q<q_c$ there is a finite-temperature phase transition; and for $q=q_c$ there is a zero-temperature critical point~\cite{Salas1998}. On the square lattice, $q_c=3$, a rigorous result~\cite{Chandgotia2018,Duminil-Copin2019,Ray2020} which tells us that the $(3,1,2)$ exclusion model is critical. 

Our class of hardcore spin models allows us to generalize these observations about Potts models in two ways. First, by considering families of fixed $\Delta$ according to the mapping in Eq.~\ref{eq:delta}, we can analyze trends in ferromagnetic and antiferromagnetic models upon approach to the XY limit. Second, by constructing a height representation we show that the 3-state Potts antiferromagnet belongs to a family of models with $p'=2$  which all have the same critical behavior, and in fact each value of $p'$ hosts a critical family.

\subsection{General Remarks about Phase Diagrams in 2D}
Having defined our new classes of models, we now situate these models within the broader context of phase transitions in spin systems.

In the absence of a temperature $T$, the phase of the system is determined instead by the tuning parameters $(N,p')$ (or $\Delta$ in the XY limit). Moreover, since all allowed configurations have the same energy, a transition at some critical angle $\Delta_c$ must be driven by the quest to maximize entropy. By analogy to the phase transition in systems of hard spheres{~\cite{Alder1957}}, which crystallize at a critical density as a means to access greater free volume per particle and increase entropy, for inclusion models we would expect ferromagnetic order to set in below $\Delta_c$, as a means to access a greater range of allowed orientations per spin.

In classical systems with short-ranged interactions in fewer than 3 dimensions, however, a continuous symmetry cannot be broken at finite temperature, by the Mermin-Wagner Theorem~\cite{Mermin1966}. In two dimensions, systems instead undergo a Kosterlitz-Thouless (KT) transition~\cite{Kosterlitz} to a continuous line of critical points. The critical phase is characterized by quasi-long-ranged order (QLRO), as correlations decay algebraically, rather than exponentially, to zero~\cite{Goldenfeld1992}:
\begin{equation}\label{eq:crit}
G(r) \equiv \langle \vec{s}(0) \cdot \vec{s}(r) \rangle \propto \frac{1}{r^{d-2+\eta}}
\end{equation}
where $G(r)$ is the spin-spin correlation function, $d=2$ is the dimensionality and $\eta$ is a critical exponent, which varies continuously with temperature in the KT phase. This follows from applying the renormalization group to a Gaussian spin-wave action with irrelevant vortex operators, which become relevant at vortex unbinding where $\eta=1/4$. 

Turning now to $Z_N$ models, for which the order parameter has a discrete symmetry, there are three potential phases: paramagnetic, critical, and ordered. A pioneering study of the clock model phase diagram considered a ferromagnetic XY model with symmetry-breaking interactions $h_N\sum_\vec{r}\cos(N\theta(\vec{r}))$~\cite{Jose1977}. These interactions become relevant at an effective coupling of $K_{eff}=N^2/8\pi$, implying that within the critical phase:
\begin{equation}\label{eq:jose}
    \frac{4}{N^2} \leq \eta \leq \frac{1}{4}.
\end{equation}
Subsequent work~\cite{Elitzur1979,Cardy1980,Frohlich1981,Borisenko2011,Kim2017,Li2020} has confirmed that for $N\geq 5$ this intermediate phase is bounded by two KT transitions, while for $N\geq 4$ the critical phase vanishes, leaving a second-order phase transition between order and disorder. Antiferromagnetic clock models are less well-studied, but it was argued that odd $N$-state antiferromagnetic models belong to the same universality class as $2N$-state ferromagnetic models~\cite{Cardy1981}.\footnote{The purported correspondence between $N$-state antiferromagnets and $2N$-state ferromagnets again comes from considering symmetry-breaking interactions with couplings $h_N$, while for pure clock models ($h_N \rightarrow \infty$), subsequent analysis challenged this interpretation~\cite{DenNijs1985}. However, we will find below that the height representation for odd $N$ exclusion models gives rise to the same lower bound of $1/N^2$ on the critical exponent $\eta$ as that obtained in Ref.~\cite{Cardy1981}.} In both ferromagnetic and antiferromagnetic models, as $N\rightarrow\infty$ we recover the XY model, and the symmetry-broken phase disappears at finite temperature.

The extent to which this canonical lore will hold for constrained systems depends on the assumption that the hard constraint becomes an irrelevant operator upon coarse graining. Although simulations have demonstrated that long-wavelength fluctuations do indeed destroy long-ranged order in 2D systems of hard disks~\cite{Illing2017}, {which instead possess an intermediate hexatic phase with quasi-long-ranged orientational order and a solid phase with quasi-long-ranged translational order~\cite{Engel2013},} the Mermin-Wagner Theorem is not guaranteed to hold for models with hard interactions~\cite{Simon1981,Halperin2019}. One extension to the theorem has come from analysis of random surface models on the two-dimensional torus, which map a set of vertices to real values $\{x\}$ subject to the nearest-neighbor potential $U(x_i-x_j)$~\cite{Mio2015}. A lower bound on the variance of fluctuations logarithmic in the length of the torus is proven for a large class of potentials including the hammock potential, defined as $U(x) = 0$ if $|x|\leq 1$ and $U(x)=\infty$ otherwise. {We have been informed that a followup work currently in preparation extends this method of proof to 2D XY models with rotationally invariant nearest-neighbor interaction, which include our classes of hardcore models of XY spins~\cite{Peled2021}.}

One clear departure of the hardcore spin models from the KT picture is in the region where vortices are explicitly forbidden. For inclusion angles below $\pi/2$, Eq.~\ref{eq:Vij} dictates that the winding number around a plaquette, and hence the winding along any closed loop which does not intersect the boundary, is zero. This amounts to setting the strength of the vortex insertion operator, which is a function of the dual field, identically to zero. This is to be contrasted with the role of defects in the typical KT phase, in which they bind into vortex-antivortex pairs up to the transition temperature~\cite{Fradkin2013}. Thus we might instead call the vortex-forbidden region a Patrascioiu-Seiler phase, after the eponymous constraint~\cite{Patrascioiu1993} applied to a 2D system of XY spins with ferromagnetic interaction in Ref.~\cite{Aizenman1994} to induce algebraic correlations at all temperatures, including $T=\infty$.
 
\section{Methods}\label{sect:methods}
This section introduces the three main methods of our study: transfer matrix analysis on semi-infinite cylinders, a Monte Carlo cluster algorithm on finite lattices of aspect ratio 1, and a height representation within the vortex-forbidden regime. Together, these methods are used to extract $\eta$, the critical exponent of the equal-time spin correlation function in the critical phase (Eq.~\ref{eq:crit}), from which we can assess how these new models fit inside our general understanding of 2D spin systems.

\subsection{Transfer Matrix Method}
For tori of length $L$ and width $W$ with the nearest-neighbor interaction defined by Eq.~\ref{eq:Vij}, the partition sum can be factored into a product over adjacent columns of $W$ spins. This is encapsulated in the $N^W \times N^W$ transfer matrix $T$, as shown in Fig.~\ref{fig:transfer}. Numbering configurations of $W$ spins by the index $i\in\{1,...,N^W\}$, $T_{ij}=1$ if there exists a valid state in which $i$ and $j$ are adjacent columns, and $T_{ij}=0$ otherwise.
\begin{figure}[hbtp]
\includegraphics[width=\linewidth]{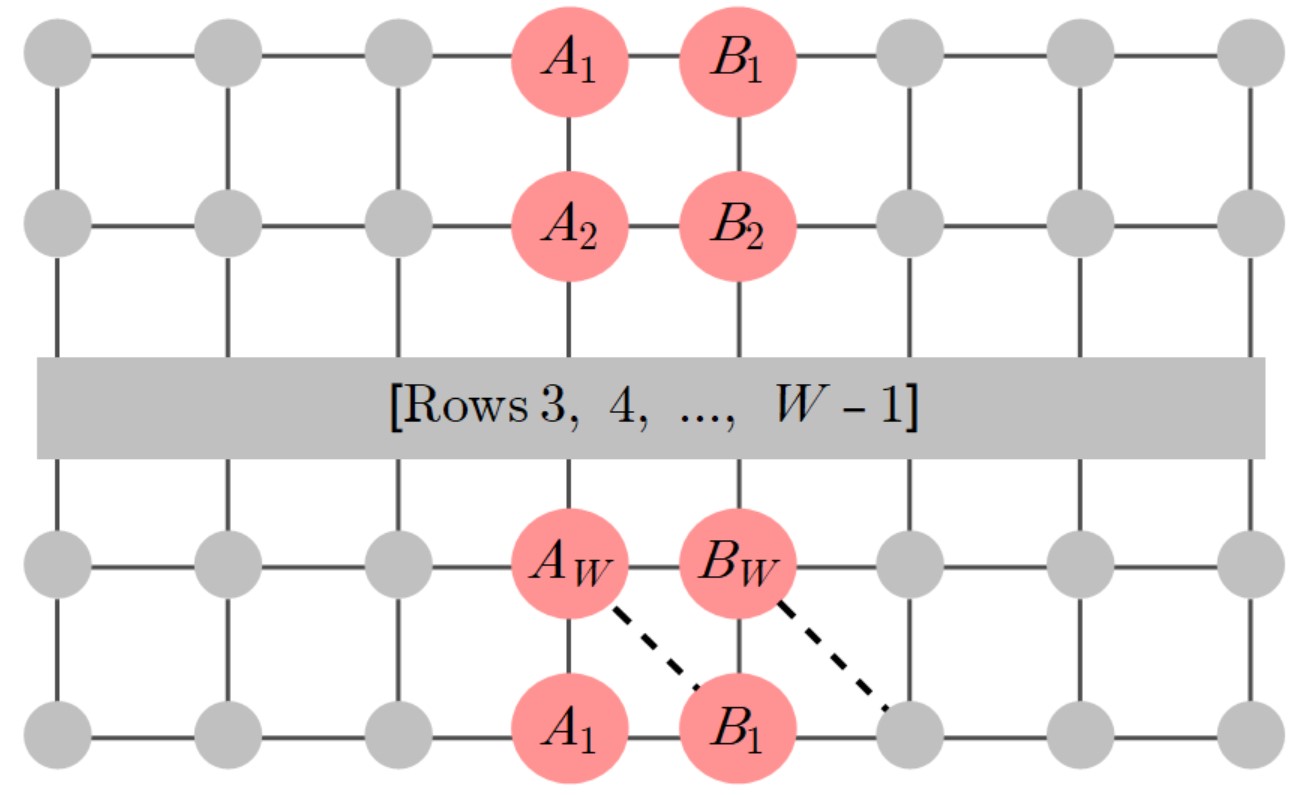}
\caption{Nearest neighbor bonds with periodic boundary conditions on a patch of a torus of width $W$. The element $T_{AB}$ is the Boltzmann weight due to the energy within and between adjacent columns in spin configurations $A$ and $B$. The bottom row is a copy of the top row for periodic boundary conditions; for twisted periodic boundary conditions, used for exclusion models when $W$ is odd, the dashed bonds are used instead.\label{fig:transfer}}
\end{figure}
The matrix $T$ depends on the boundary conditions applied in the transverse direction. We use periodic boundary conditions (PBCs) in all cases except for exclusion models at odd widths. In that case, the periodic boundary conditions are shifted by one lattice constant to make the lattice bipartite, thus circumventing the domain wall of infinite energy cost that would arise in an antiferromagnetic state with regular PBCs. This allows us to extract meaningful data for all numerically feasible widths.

We numerically compute the largest eigenvalues of $T$. If the eigenvalues are nondegenerate, then in the limit $L\rightarrow \infty$, the free energy per ring is given by~\cite{Goldenfeld1992}:
\begin{align}\label{eq:f}
f &= -\frac{1}{\beta L} \log \mathrm{Tr} T^L = -\frac{1}{\beta L}\log \sum_\Lambda \Lambda^L \notag \\
&\rightarrow -(1/\beta) \log |\Lambda_1|
\end{align}
where $\Lambda_1$ is the eigenvalue of maximum absolute value.
For a hard potential, the partition sum is taken over zero-energy configurations only, so from the transfer matrix we can also directly read off the entropy per ring on an infinite cylinder:
\begin{equation}\label{eq:s}
s = k_B\log|\Lambda_1| = \lim_{L\rightarrow\infty} \log (\Omega/L).
\end{equation}

The transfer matrix method diagnoses the critical phase via the well known mapping between the infinite plane and the cylinder in CFTs~\cite{Cardy1985}. Recall that discretizing a model on a long, narrow cylinder of finite width and infinite length makes it effectively one-dimensional. This destroys (quasi) long-ranged order, so that, for a system which in two dimensions is critical, Eq.~\ref{eq:crit} is modulated by an exponentially decaying function whose correlation length is given by the gap between the two largest eigenvalues~\cite{NightingaleP;Blote1983}:
\begin{equation}\label{eq:cylinder}
G(r) \sim r^{-\eta} e^{r/\xi} \quad \mathrm{where} \quad \xi = (\log |\Lambda_1/\Lambda_2|)^{-1}.
\end{equation}
The conformal mapping shows that $\xi$ scales linearly with width:
\begin{equation}\label{eq:A}
\xi(W) = W/A \quad \mathrm{where} \quad A=\pi \eta.
\end{equation}
Computational constraints limit us to small $N$ and $W$, for which linear fit to $\xi(W)$ tends to overestimate $\eta$. 
Alternatively, we can calculate a width-dependent exponent $\eta(W)$ by plugging data at different widths into Eq.~\ref{eq:A}, and then performing a fit to determine $\eta_\infty$: 
\begin{equation}\label{eq:etainf}
\eta(W) = \eta_\infty + C/W^2
\end{equation}
which tends to underestimate $\eta$.
Therefore, we mainly use the transfer matrix method as a tool to identify critical models, for which a more accurate estimate of the critical exponents can be determined using the cluster algorithm detailed in the next section. For finite lengths, the correlation function measured by the cluster algorithm can be checked against $G(r)$ as determined by explicit multiplication of the transfer matrix. Most important for our purposes is Eq.~\ref{eq:s}, by which we identify families of models with the same macroscopic entropy density.

\subsection{Cluster algorithm}
To access larger system sizes, and the XY limit, we use a Monte Carlo algorithm that remains ergodic in the critical phase. An algorithm with single-spin updates does not fit these criteria, since it cannot unwind defects, and the acceptance rate of single-spin moves becomes prohibitively low as the inclusion angle decreases. To probe configuration space efficiently, a cluster algorithm is required instead. To do so, we need to specify (1) how the cluster is constructed and (2) what operation is applied to this cluster.
 
A central ingredient in answering question (1) is the concept of pocket Monte Carlo, which constructs a pocket, or bag of elements (in our case, $Z_N$ or XY spins), to which the cluster operation is applied~\cite{Krauth2003}. We remove an item from the pocket by applying the operation to it, which then causes other elements (e.g., the neighbors of the element just transformed) to be added to the pocket. The process terminates when the pocket is empty.

\begin{algorithm}[hbtp]
\SetAlgoLined
\SetKwInOut{Output}{Output}
\SetKwInOut{Input}{Input}
\Input{Configuration of spins $\{\vec{s}_i\}_{i=1}^{N_s}$ at iteration $n$}
\Output{Configuration of spins $\{\vec{s}_i\}_{i=1}^{N_s}$ at iteration $n+1$}
$\mathcal{P} \gets$ empty bag\;
$j \gets$ random integer from 0 to $N_s-1$ \;
add $\vec{s}_i$ to $\mathcal{P}$ \;
$\theta \gets$ random floating point number from 0 to $2\pi$\;
\While{$\mathcal{P}$ is not empty} {
remove one spin $\vec{s}$ at random from $\mathcal{P}$\;
$\vec{s} \gets$ reflect($\vec{s}, \theta$) \;
\For{$\vec{s}_k$ \normalfont{\textbf{in} neighbors of} $\vec{s}$} {
	\If{$\mathrm{angle}(\vec{s}_k, \vec{s})$ \normalfont{violates hard constraint}}{
	add $\vec{s}_k$ to $\mathcal{P}$ \;
}
}
}
\Return $\{\vec{s}_i\}_{i=1}^{N_s}$ \;
\caption{Reflect\label{alg:cluster}}
\end{algorithm}

The operation applied to the cluster in Algorithm~\ref{alg:cluster} is a reflection about the axis $\theta$. The move begins by reflecting one spin, $\vec{s}_j$, about this axis and builds the cluster with spins that violate the hard constraint (line 9). This takes inspiration from the pivot cluster algorithm for hard disks~\cite{Dress1995} and pocket dimer algorithm for non-overlapping dimers~\cite{Krauth2002}, in which the items in the pocket are reflected about a chosen axis or point, and overlapping disks/dimers are added to the pocket. However, while these pivot algorithms implement a reflection in real space, the \texttt{reflect} algorithm performs the transformation in spin space. In this sense it is more akin to the Wolff algorithm~\cite{Wolff1989}, famous for its application to the 2D Ising Model but also generalizable to XY spins~\cite{Luijten2006}. {Previous works have used this algorithm to simulate XY~\cite{Bietenholz2013,Bietenholz2013a} and $O(3)$~\cite{Hasenbusch1996} inclusion models, but Algorithm~\ref{alg:cluster} can also be applied to exclusion and inclusion clock models, simply by modifying line 4 to sample from a discrete set of reflection axes.}

Since the angles between each pair of spins in the cluster are preserved, it is straightforward to prove that the \texttt{reflect} algorithm respects detailed balance. Moreover, for the models under consideration the algorithm is ergodic, as proven in Appendix ~\ref{sect:app1}. A given spin is added to the pocket at most once, so the outer loop is guaranteed to terminate after $\leq N_s$ iterations, where $N_s$ is the number of spins. In the worst case, the cluster spans the entire lattice, and the move is a global reflection. This is a problematic aspect of the pivot cluster algorithm when applied to systems of monodisperse disks: the algorithm, which becomes inefficient when clusters are too large, has a percolation threshold below the liquid-solid transition density~\cite{Krauth2003}. Thus we should be wary of a similar issue occurring for hardcore spins. However, we instead find that the scaling of cluster sizes closely tracks the critical properties of the underlying spin system. Namely, as with the Wolff algorithm applied to Potts models, the average cluster size scales as $\langle s \rangle \propto L^{2-\eta'}$, where $\eta' \approx \eta$, with a relative error $\lesssim 0.01$ for inclusion angles near $\pi/2$. The potential exactness of this relation is explored further in Appendix~\ref{sect:app2}.

The critical exponent $\eta$ is determined from the Fourier spectrum of the spin correlation function, denoted $G(\vec{k})$. This spectrum has a peak at $\vec{k} = (0, 0)$ for inclusion models, corresponding to the uniform susceptibility $\chi_u$, and $\vec{k} = (\pi,\pi)$ for exclusion models, corresponding to the staggered susceptibility $\chi_s$. Given a model whose real-space correlation function decays as a power law (Eq.~\ref{eq:crit}), the susceptibility on an $L\times L$ lattice scales as $L^{2-\eta}$, times sub-leading logarithmic corrections~\cite{Hasenbusch2005}. Performing a linear fit to $\log \chi$ vs. $\log L$ provides an estimate of $\eta$. For small system sizes the best fit for $\eta$ is sensitive to the minimum system size $L_{min}$ included in the fit, but we find that the estimated value is roughly stable for $L_{min}\geq 24$, and therefore start all fits at $L=24$. 
By contrast, in the paramagnetic phase where $G(r)$ decays exponentially to zero, $\chi(L)$ approaches a plateau at a scale set by the correlation length.

We employ a variety of checks to verify that the algorithm converges to the stationary distribution of the system. First, to check that the system decorrelates from its initial state, we run the algorithm from different initial configurations and measure the standard deviation between runs. To generate disordered initial states, we soften the potential into a continuous, differentiable function of the nearest-neighbor angle. The system is initialized in a random configuration, which is likely to have nonzero energy due to violated constraints. The initialization algorithm then alternates between (1) gradient descent to the nearest minimum of the soft potential and (2) a full sweep of zero-temperature single-spin moves, which are accepted only if they reduce the energy. This process is repeated until a zero-energy state is found, which then becomes the initial state for the  cluster algorithm, or, if no zero-energy state is found after a maximum number of attempts, the ordered state is used instead. For $Z_N$ models, a differentiable soft potential cannot be defined, so we instead initialize in one of the ideal states introduced below. We also use the ideal state initialization for large systems of XY spins at low inclusion angles ($\Delta < \pi/2$) for which the minimization algorithm generally fails to find a zero-energy state.

A second check is that the spatial correlations determined for $Z_N$ models on tori of small widths should match those determined from the transfer matrix. This comparison was performed for three different models on tori of width 4 and aspect ratios 1, 2, 4, and 8, with a maximum relative error of $2.54\times 10^{-3}$ for $G(r)>0.005$.

To control for the effects of critical slowing down, a single Monte Carlo step was defined as follows. First, we defined one MCS as a single cluster move and calculated the autocorrelation function of the susceptibility. The integrated correlation time was conservatively estimated to scale with system size $L$ no faster than $\sim L^{1/2}$. For comparison, this is not as efficient as the Wolff algorithm for the Ising Model, which scales logarithmically or with a very small dynamical exponent near the critical point~\cite{Baillie1990,Du2006}, but it is a vast improvement on the $O(L^2)$ critical slowing down observed for local algorithms ~\cite{Krauth}. One Monte Carlo step was then redefined to consist of $C L^{1/2}$ cluster moves, with $C$ chosen such that on the smallest system size considered (8x8), 1 MCS=8 cluster moves. Through this redefinition, we could then generate at least as many independent samples with increasing system size by running the algorithm for the same number of steps.

Each run consisted of $10^4$ MCS of equilibration followed by $10^5$ MCS of recording. The correlation time was (much) less than 100 MCS for all models and system size considered, so all error bars were estimated by averaging the data over chunks of 100 MCS and computing the standard error over 1000 chunks. A data bunching analysis verified that these chunks were uncorrelated~\cite{Krauth}.

\subsection{Height representation for Defect-Free Models}
To complement this algorithmic approach, we formulate a height representation for $Z_N$ models with $N>2(p'-1)$. This condition forbids defects in the height field, which in the case of inclusion models, can be interpreted simply as forbidding vortices and antivortices. In the height representation, also referred to as an interface model, the spin variables are mapped to coarse-grained height variables, an approach that has been used to study a range of critical ground state ensembles~\cite{Zeng1997,Burton1997} including the 3-state Potts AFM on the square lattice~\cite{Salas1998} and the 4-state Potts AFM on the triangular ~\cite{Moore2000} and Kagom\'e lattices~\cite{Huse1992}.

For our class of inclusion and exclusion models, we map from clock variables $\{\sigma\}=\{0,1,...,N-1\}$ to heights $\{h\}$ as follows:
\begin{itemize}
\item At the origin, define
\begin{equation}\label{eq:origin}
h(0) = \sigma(0).
\end{equation}
\item The local change in the height field from site $x=(x_1, x_2)$ to nearest neighbor site $y$ is:
\begin{equation}\label{eq:modN}
h(x) - h(y) = \begin{cases} (\sigma(x) - \sigma(y)) & \mathrm{inclusion} \\ 
(\sigma(x) - \sigma(y) - N/2) & \mathrm{exclusion} 
\end{cases}\mod N
\end{equation}
with the modulus chosen such that, in both cases,
\begin{equation}\label{eq:pm1}
|h(x) - h(y)| \leq (p'-1)/2.
\end{equation}
\item  For $N>2(p'-1)$, $\Delta h = 0$ around a plaquette; this is what is meant by the absence of defects. Only if this condition is satisfied can $h$ be uniquely defined:
\begin{equation}\label{eq:mod}
h(x) = \begin{cases}
\sigma(x) & \mathrm{inclusion} \\
\sigma(x) - \alpha(x)N/2 & \mathrm{exclusion}
\end{cases} (\mathrm{mod \,} N)
\end{equation}
where $\alpha(x) \equiv (x_1+x_2) \mod 2$ is 0 on the even sublattice and $1$ on the odd sublattice. By the notation (mod $N$) we mean that if the modulo operation is taken on both sides, then the equality holds. For inclusion models, the distinction between $h$ and $\sigma$ has a simple interpretation: while the spin state $\sigma$ is only defined modulo $N$, $h$ has been "lifted" to $\mathbb{Z}$~\cite{Cardy}. 
\end{itemize}
{We note in passing that the above construction for inclusion models has also appeared in the mathematical literature~\cite{Peled2017a, Peled2017}, where height fields with $p'=2m-1$ are known as $m$-Lipschitz functions. These works primarily focus on high dimensions, where long-range order can be proven in the vortex-forbidden regime.}

Returning to $d=2$ and a more physically motivated perspective, given a uniquely defined height field $h$, we posit that it is governed by the effective Hamiltonian:
\begin{equation}\label{eq:effective}
F = \int d^2 x \left(\frac{K}{2}(\nabla h)^2 + V_{lock}(h)\right)
\end{equation}
where $K$ is the stiffness and $V_{lock}$ is the so-called locking potential. The latter favors heights on the ideal state lattice $\mathcal{I}$, the set of periodic, macroscopically flat height configurations with maximal entropy (ideal states)~\cite{Kondev1995b,Burton1997}. 

In writing down the effective Hamiltonian, a typical ground state is assumed to be a patchwork of ideal state domains. We refer the reader to Ref.~\cite{Salas1998} for a detailed review of the terminology and briefly recapitulate the main definitions here. Two instances of the same domain $X$ have equal heights modulo $\tilde{h}$, where $\tilde{h}$ is an element of the repeat lattice $\mathcal{R}$. The repeat lattice is a subgroup of the equivalence lattice $\mathcal{E}$, the set of elements $a$ such that $a+\mathcal{I}=\mathcal{I}$. Thus, $V_{lock}$ has the same periodicity as the equivalence lattice, which in our case has a one-dimensional representation with period $2\pi/g_0$.

If the locking potential is relevant ($\eta_{lock} < 2d=4$), the model is said to be in the smooth phase, with long-ranged order in the height field. On the other hand, if $\eta_{lock}>2d$, the locking potential is irrelevant and the interface is rough, with logarithmic correlations in the height field:
\begin{equation}\label{eq:hcorr}
    G_h(x,y) \equiv \langle(h(x)-h(y))^2\rangle \sim \frac{1}{\pi K}\log|x-y|.
\end{equation}
The critical exponents associated with relevant vertex operators in the original spin model can then be determined by expressing them as periodic functions of the height field and applying Eq.~\ref{eq:hcorr}. Due to the Gaussian form of the action, an operator $O$ with period $2\pi/g$ scales as:
\begin{align}\label{eq:gaussian}
    \langle O^*(x) O(y) \rangle &= \langle \exp[i g (h(x)-h(y))]\rangle \notag \\
    &= \exp[-G_h(x,y) g^2 /2]
\end{align}
implying a critical exponent: 
\begin{equation}\label{eq:eta-K}
    \eta_O = \frac{g^2}{2\pi K}.
\end{equation}
The roughening transition therefore occurs when $g_0^2/2\pi K = 4$.

To determine the critical exponent of the spin correlation function, denoted simply as $\eta$, we express the magnetization in terms of the height field as:
\begin{equation}\label{eq:unstag}
e^{\pm 2\pi i h(x)/N} = M_1(x) \pm i M_2(x)
\end{equation}
where $\vec{M}=(M_1, M_2)$ is the unstaggered magnetization for inclusion models, and the staggered magnetization for exclusion models. Eq.~\ref{eq:eta-K} then implies:
\begin{equation}\label{eq:K}
\eta = \frac{2\pi}{N^2 K}.
\end{equation}

For inclusion, ideal states consist of spins randomly sampled from a set of $p'$ consecutive spin states. Equivalently, for even $N$ exclusion models, we choose a set $S$ of $(p'+1)/2$ consecutive integers as the spin states on the even sublattice and populate the odd sublattice with the integers $(S + N/2) \mod N$ (i.e., rotate the spins on the even sublattice by $\pi$), so per Eq.~\ref{eq:mod}, the average height is constant and equal on the two sublattices. In both cases, there are $N$ ideal states, in one-to-one correspondence with the average height mod $N$. The magnetization has the periodicity of the repeat lattice $\mathcal{R}=N\mathbb{Z}$, whereas the locking potential has the periodicity of the equivalence lattice $\mathcal{E}=\mathbb{Z}$. Thus, for the system to occupy the rough phase, this implies the relation:
\begin{equation}\label{eq:lock}
4 < \eta_{lock} = \frac{2\pi}{K} = N^2 \eta \qquad (\mathrm{odd \,} p').
\end{equation}
The lower boundary of the critical phase, $\eta=4/N^2$, is consistent with Inequality~\ref{eq:jose} derived for ferromagnetic XY models with clock perturbations.

For exclusion models with even $p'$ (odd $N$), the absence of $\mathbb{Z}_2$ symmetry leads to an asymmetry in the ideal states: choose a set $S = \{\sigma_i\}$ of $p'/2$ consecutive clock orientations as the allowed states on one sublattice, and a set $T$ of $p'/2 + 1$ integers on the other sublattice: $T= (S + (N-1)/2) \cup \{\sigma_{p'/2} + (N+1)/2\}$. Owing to this asymmetry, there are now $2N$ ideal states, uniquely identified by their average heights mod $N$. Referring back to Eq.~\ref{eq:mod}, note that $h(x)$ takes integer values on the even sublattice, and half-integer values on the odd sublattice. The ideal state lattice and equivalence lattices are $\mathbb{Z}/2$, whereas the repeat lattice (which again has the same periodicity as $\vec{M}$) is $N\mathbb{Z}$. This implies that within the rough phase:
\begin{equation}\label{eq:lock2}
4 < \eta_{lock} = \frac{8\pi}{K} = 4N^2 \eta \qquad (\mathrm{even \,} p').
\end{equation}
Therefore, in contrast with Eq.~\ref{eq:lock}, the lower boundary of the critical phase is characterized by $\eta = 1/N^2$. {This agrees with the bound on the critical phase derived from a renormalization group treatment of odd $N$ antiferromagnetic clock models with standard action~\cite{Cardy1981}. Interestingly, the RG arguments applied exclusively to $T\neq 0$, whereas the height representation allows us to derive the same bound at $T=0$, for which these clock models reduce to $p'=2$ exclusion models.} 

\comment{
Another difference from even $N$ is, whereas for even $N$ the ideal states have zero net magnetization, for odd $N$ ideal states do have net magnetization. The magnetization is determined by the choice of $S$, independent of which sublattice is populated by $S$. The height field satisfies:
\begin{equation}
h(x) = \frac{x_1 + x_2}{2} \mod 1
\end{equation}
so:
\begin{equation}
e^{\pm 2\pi i h(x)} = (-1)^{x_1 + x_2}
\end{equation}
which, combined with Eq.~\ref{eq:stag}, implies:
\begin{equation}
e^{\pm 2\pi i (N-1)/N h(x)} = M_1(x) \mp i M_2(x)
\end{equation}
}
 
 To summarize, the existence of a height representation indicates that all inclusion and exclusion models with $N>2(p'-1)$ are either critical or ordered. The critical exponents of models in the critical phase must satisfy either Eq.~\ref{eq:lock} or Eq.~\ref{eq:lock2}. This leaves open two questions: whether any models fall within the smooth phase (long-ranged order), and whether there are critical models which do not admit a height representation. 
 
 \section{Phase Diagram}\label{sect:results}
With the height representation as a guide, we now describe in detail the phase diagram of $Z_N$ models (Fig.~\ref{fig:phase}), using the \texttt{reflect} algorithm to determine the critical exponents. Models described by a height representation are classified into $p'$ families, all in the rough phase except for $p'=1$. For small $p'$ the transition to the paramagnetic phase is geometric, coinciding with the point at which defects are allowed, but for large $p'$ and  in the XY limit, we identify a region of the phase diagram in which vortices are energetically allowed but entropically disfavored, leading to quasi-long-ranged order by disorder. The transition from the critical phase to the paramagnetic phase is a Kosterlitz-Thouless-type transition, driven by the unbinding of vortices at $\eta=1/4$.

\subsection{Families of $Z_N$ Models}
\begin{figure}[t]
    \centering
    \includegraphics{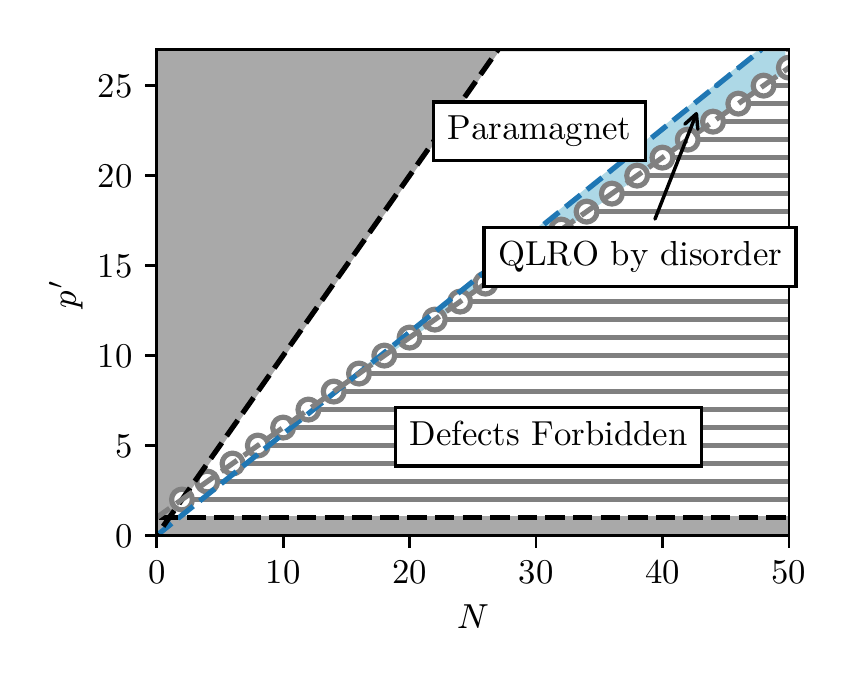}
    \caption{Phase diagram in the $(N,p')$ plane for $Z_N$ inclusion and exclusion models. $p'$ families are represented by gray horizontal lines, with open circles along the gray dashed line denoting $N=N_c(p')$. Dark gray shaded regions are forbidden, and their black dashed boundaries are the trivial limiting cases: $p'=1$ is strictly ordered with zero macroscopic entropy density, while $p'=N$ is disordered (unconstrained/noninteracting). Constant inclusion angle corresponds to functions of the form $p' = N\Delta/\pi$; the blue dashed line has $\Delta=\Delta_c\approx 0.56\pi$ and marks the upper boundary ($\eta=1/4$) of the quasi-long-ranged order by disorder phase (Eq.~\ref{eq:pp2}). \label{fig:phase}}
    \end{figure}
    We begin by assessing trends in the transfer matrix, from which the following definition arises:
\begin{definition}[$p'$ family]\label{def:pp}
A class of models with a common $p'$ and a critical value $N_c(p')$ such that all models with $N > N_c(p')$ have the same maximum transfer matrix eigenvalue for each width.
\end{definition}
Each family has an assortment of scaling properties. On a torus of finite length, the number of configurations can be factored as:
\begin{equation}
\Omega(N,W,L) = N\exp(S(W,L)).
\end{equation}
That is, for a fixed width $W$ and length $L$, $\Omega$ scales proportionally to $N$. Due to the use of longitudinal periodic boundary conditions implicit in Eq.~\ref{eq:f}, this equation only strictly holds for small lengths, since longer strips can support nontrivial windings. 

While Definition~\ref{def:pp} classifies $p'$ families in terms of their eigenvalues in the spin representation, members of the same family are related more fundamentally through their height representation. On the square lattice,
\begin{equation}\label{eq:Nc}
N_c(p') = 2(p'-1).
\end{equation}
Therefore, the condition $N>N_c(p')$ within a $p'$ family---which for inclusion models amounts to forbidding vortices---coincides with to the criterion which allows the height field to be uniquely defined. Defining $\Delta$ to saturate the right-hand-side of Inequality~\ref{eq:delta}, Eq.~\ref{eq:Nc} corresponds to an inclusion angle of $\pi/2$ (gray dashed line in Fig.~\ref{fig:phase}). Below this angle, members of the same $p'$ family admit a height representation obeying the same local height rule (Inequality~\ref{eq:pm1}), and thus the same stiffness.

The value of the stiffness determines whether the family occupies the smooth or rough phase. In the latter phase, Eq.~\ref{eq:K} implies that $\eta N^2$ is invariant. Since $\xi(W) \propto W/\eta$ at criticality on semi-infinite cylinders of width $W$, the correlation length factorizes as:
\begin{equation}\label{eq:corr-pp}
\xi(W, N, p') \approx N^2 \tilde{\xi}(W, p').
\end{equation} 
It must be emphasized that this scaling form is only approximate for finite widths. For, the universal quantity measured in different members of the same family is the correlation function of the height field, $G_h(x-y)$ defined in Eq.~\ref{eq:hcorr}. This can only be related to a correlation in the spin language if the Gaussian action of the form Eq.~\ref{eq:effective} is assumed. In that case, positing exponential decay on the left hand side of Eq.~\ref{eq:gaussian} in the quasi-1D limit immediately yields Eq.~\ref{eq:corr-pp}. But this is a coarse-grained, long-wavelength description of the system, and should not be expected to hold exactly for small widths. 
Indeed, as seen in Fig.~\ref{fig:pp} for the $p'=2$ family, $\xi/N^2$ at a given width is a monotonically increasing function of $N$, although this scaling correction becomes less severe as $W$ increases.

The linear trend with width exhibited in Fig.~\ref{fig:pp} is indicative of the critical phase. But just as the data collapse between members of the same family improves with $W$, so too does the critical exponent determined from a linear fit to the combined data becomes more accurate as we access larger widths. The largest feasible width is limited by computational constraints, which thus limits the accuracy of our estimate of $\eta$.
\begin{figure}[t]
    \centering
    \includegraphics{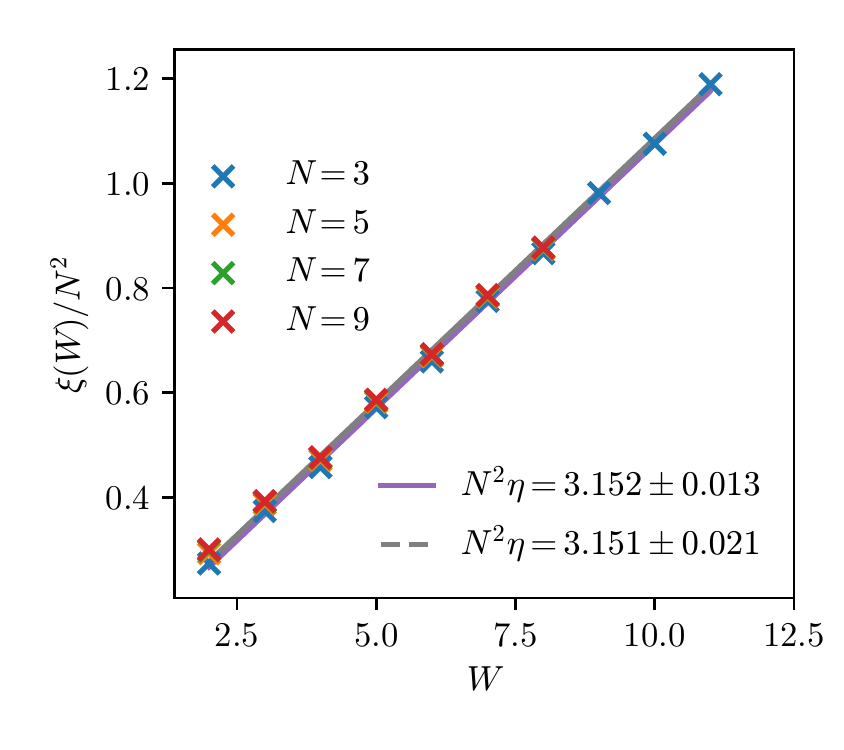}
    \caption{$\xi(W)/N^2$ for the $p'=2$, as determined from the transfer matrix for small values of $N$. Different-colored x's correspond to different members of the family, while the solid (dashed) lines show linear fits for even (odd) widths. Very small widths ($W<4$) were not included in the fits. \label{fig:pp}}
\end{figure}
For the $p'=2$ family, a linear fit of the form $\xi(W)/N^2 = B W + C$ to data points with $W\geq 4$ for even (odd) widths yields $\eta N^2 = 3.152 \pm 0.013$ ($\eta N^2 = 3.152 \pm 0.021$). Alternatively, a fit of the form Eq.~\ref{eq:etainf} yields $\eta N^2 = 2.994 \pm 0.3$ ($\eta N^2 = 2.996 \pm 0.45$). These fits are broadly consistent with the known exponent of $\eta=1/3$ for the staggered susceptibility of the zero-temperature 3-state Potts AFM~\cite{Burton1997}, which belongs to the family. 

This not only provides a useful cross-check of the transfer matrix method, but also points to the power of the height representation: we can immediately bootstrap exact knowledge of the stiffness or approximate determination of $\eta$ for just one member of the family to a theoretical prediction of $\eta$ for \textit{all} members of the family. In this sense the $p'=2$ family (or, put differently, the ground-state ensembles of odd $N$ clock models with standard antiferromagnetic cosine action) can be said to descend from the 3-state Potts AFM at zero temperature, whose critical exponent places the family safely in the rough phase (Eq.~\ref{eq:lock2}, which crucially must be distinguished from Eq.~\ref{eq:lock} for odd $p'$).

{We will demonstrate below that all $p'>1$ families occupy the rough phase, and thus the transfer matrix also enables a crude estimate of the central charge $c$ of the associated conformal field theory. For $p'=2$ the central charge can be deduced from the 3-state Potts AFM at zero temperature, for which $c=1$ is known exactly~\cite{Shrock1998}. In fact, through finite size scaling of the entropy density, we find $c\approx 1$ for all families examined. This matches our interpretation of $p'$ families as interface models governed by the effective action in Eq.~\ref{eq:effective}, which if the locking potential is irrelevant is a free boson CFT with $c=1$~\cite{Qualls2015}.} The same central charge was determined numerically for clock models with $5\leq N \leq 8$ obeying the standard action~\cite{Li2020}.
\comment{
Before calculating these critical exponents more precisely using the \texttt{reflect} algorithm, we use the transfer matrix to briefly examine another sense in which this class of exclusion and inclusion models can be viewed as a generalization of Potts models. Instead of families with fixed $p'$, consider models with fixed $q=N/p_{ex}$. The number of configurations on a torus of finite length and width is found to factorize as:
\begin{equation}
\Omega(q, p', W, L) = (p')^{WL} f(W, L)
\end{equation}
The striking feature here is that all dependence on $q$ is absorbed into the definition of $p'$. This equation is motivated by the fact that, starting from an antiferromagnetic configuration, each spin has an average "wiggle room" proportional to $p'$. One can think of this in terms of the ideal states of the height representation (even when defects are allowed and the height field cannot be uniquely defined). So the exponential of the entropy density in the limit $L\rightarrow\infty$ on a torus of width $W$ scales as:
\begin{equation}\label{eq:entropy}
e^{s(W)} = |\Lambda_1|^{1/W} = p' \lim_{L\rightarrow\infty} f(W,L)^{1/WL}
\end{equation}
Data from different $q$ collapse onto the same linear trend as a function of $p'$ \gs{I have a figure to this effect which might not be worth including}. In contrast, the correlation length for a given width decreases monotonically with $q$, and scales roughly as $C-A/p'$.}

\subsection{Kosterlitz-Thouless Transition}
To determine $\eta$ with greater accuracy, the finite size scaling of the susceptibility is measured using the \texttt{reflect} algorithm. We begin by considering inclusion models in the XY limit, as defined by Eq.~\ref{eq:Vij}, and provide evidence of a Kosterlitz-Thouless transition at critical angle $\Delta_c$. {Some results in this section overlap with and confirm the findings of Refs.~\cite{Bietenholz2013,Bietenholz2013a} but are included for completeness, to contextualize our findings on a broader class of models.}

As shown in Fig.~\ref{fig:FSS}, the power-law scaling $\chi \sim L^{2-\eta}$ remains a good fit to the susceptibility, with no obvious trend in the residuals, up to $\Delta \approx 0.57\pi$. For larger inclusion angles, the fit becomes increasingly poor; rather than scaling linearly with $\log L$, $\log\chi$ becomes a concave down function of $L$ indicative of the paramagnetic phase. The persistence of quasi-long-ranged order past the angle at which vortices become allowed ($\Delta=0.5\pi$) lends support to the interpretation that the system becomes a paramagnet when defects become relevant and vortices unbind, the same mechanism which drives the KT transition in the XY model at finite temperature. 

To verify this interpretation and pinpoint the critical angle $\Delta_c$ at which this transition occurs, we consider three key observables near the apparent transition: the critical exponent $\eta$, the second moment correlation length $\xi_{2nd}$, and the Binder cumulant $U$. These are discussed in turn.

\begin{figure}[t]
    \centering
    \includegraphics[width=\linewidth]{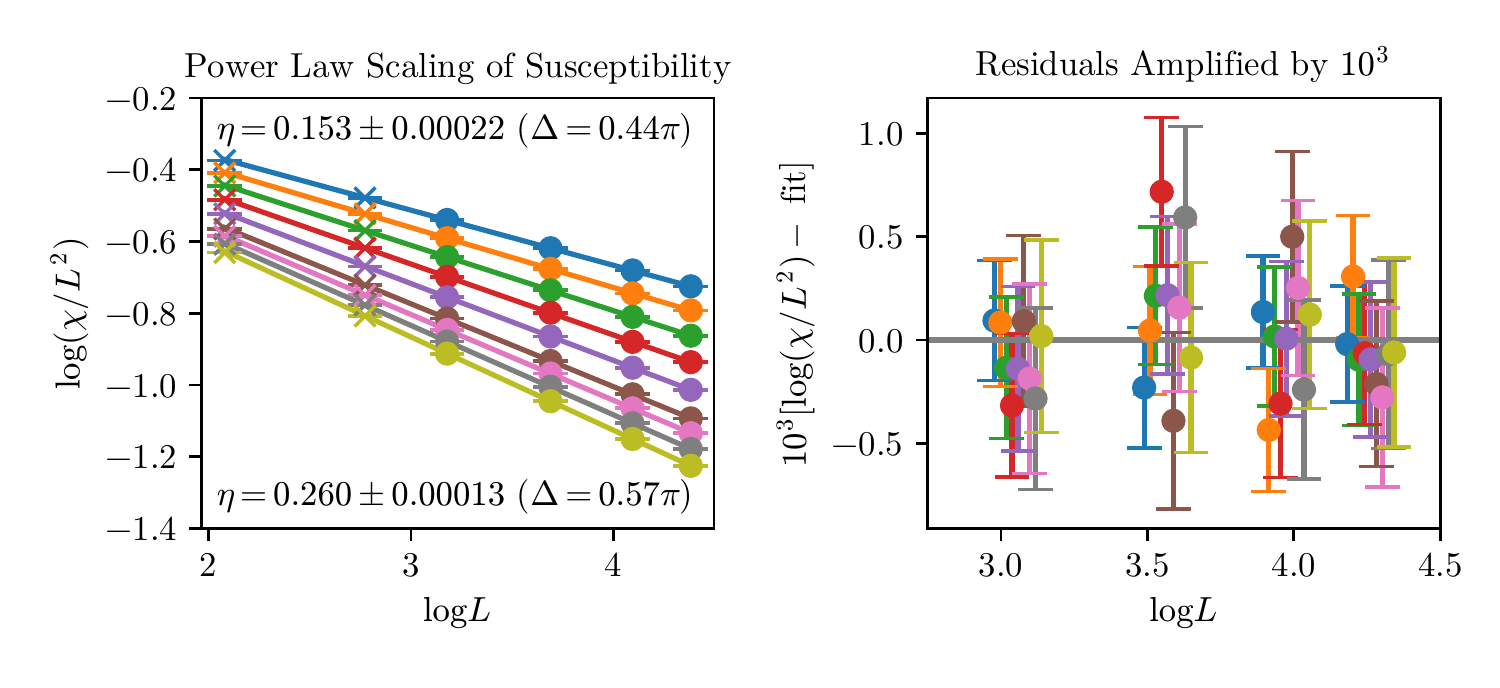}
    \caption{Power-law scaling of the susceptibility. The left panel shows $\log (\chi/L^2)$ as a function of $\log L$ for $L=8,16,24,40,60,80$. Data from $L=8$ and $L= 16$, marked with x's, were not included in the linear fit. From top to bottom, the inclusion angles are $0.44\pi$, $0.46\pi$, $0.48\pi$, $0.5\pi$, $0.52
    \pi$, $0.54\pi$, $0.55\pi$, $0.56\pi$, and $0.57\pi$.  The right panel shows residuals with respect to the fits, amplified by a factor of $10^3$ and with data from different inclusion angles scattered in the horizontal direction to enhance visibility. Error bars are defined as the standard error across 1000 consecutive bunches of 100 MCS each. \label{fig:FSS}}
    \end{figure} 
\subsubsection{Critical exponents $\eta$ and $\eta'$}
Standard KT theory indicates that the most relevant defect operator has $\eta_{vortex}=1/\eta$, so vortex unbinding occurs at $\eta=1/4$. Therefore, if the transition in our class of inclusion models is in the KT universality class, we expect:
    \begin{equation}\label{eq:etaKT}
        \eta(\Delta=\Delta_c) = 1/4.
    \end{equation}

The threshold value of $1/4$ is indicated by the gray line in Fig.~\ref{fig:KTeta}, implying a critical angle just above $\Delta=0.56\pi$. This is roughly consistent with the above finding that the power law scaling of the susceptibility holds up to $\Delta \approx 0.57\pi$.

Fig.~\ref{fig:KTeta} contains two other noteworthy features. First, the exponent $\eta'$ determined from finite size scaling of the average cluster size, closely tracks $\eta$. This is a positive indication of the algorithm's efficiency and is addressed in Appendix ~\ref{sect:app2}.

\begin{figure}[t]
    \centering
    \includegraphics[width=0.85\linewidth]{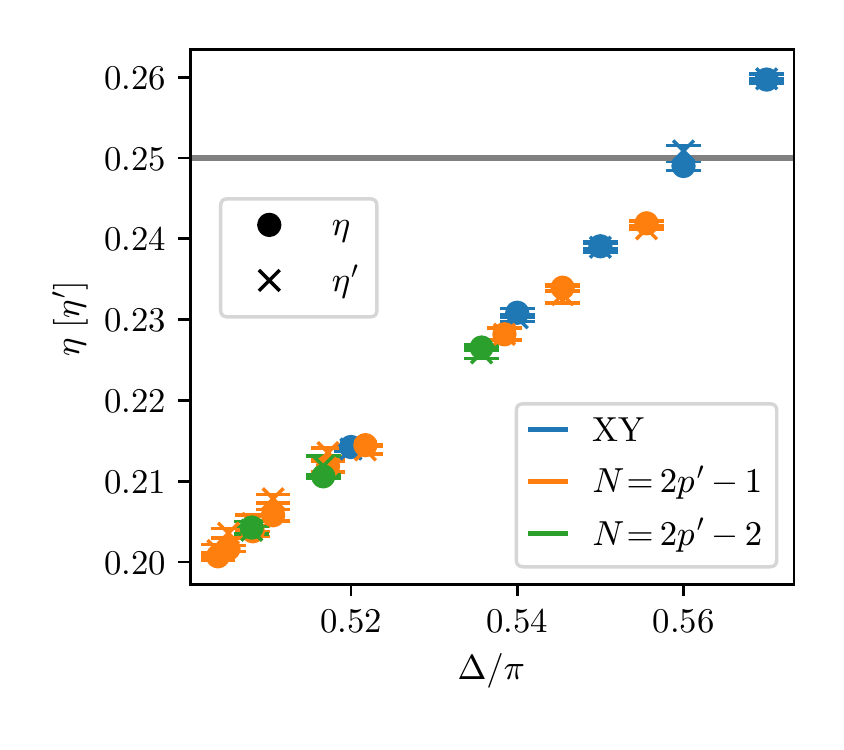}
    \caption{Critical exponents approaching the KT transition. The ray line in the left panel marks the value of $\eta=1/4$ at vortex unbinding. The critical exponents $\eta$, determined from finite size scaling of the susceptibility, and $\eta'$, determined from finite size scaling of the average cluster size, are marked with circles and x's respectively. The plotted error bars show the statistical uncertainty of the log-log fits in Fig.~\ref{fig:FSS}. Different colors denote models in the XY limit (blue), the smallest members of $p'$ families, with $p'$ ranging from 5 to 63 (orange), and clock models just outside the $p'$ families, i.e. $N=N_c$, with $p'$ ranging from 15 to 63 (green).}
    \label{fig:KTeta}
\end{figure}
Second, the critical exponents associated with various clock models, measured by the same method as in Fig.~\ref{fig:FSS}, are plotted on the axes $\eta$ vs. $\Delta/\pi$ by defining $\Delta$ at the midpoint of the interval in Inequality~\ref{eq:delta}, $\Delta=\pi p'/N$. Two sets of models are shown. Marked in orange are models for which $N=2p'-1$, which are the smallest members of their respective $p'$ families and are therefore known rigorously to be critical based on the existence of a height representation. Marked in green are those with $N=N_c(p')=2p'-2$, which lie just outside the $p'$ families. Both sets of models approach an inclusion angle of $0.5\pi$ from above in the limit $N\rightarrow\infty$. Strikingly, even for finite $N$, the two sets of models collapse onto roughly the same trend when plotted vs. $\Delta$, and this trend is also consistent with that of the models in the XY limit. The existence of critical models outside the $p'$ families again points to the phenomenon of quasi-long-ranged order by disorder, as the $N=N_c(p')$ models exhibit power law scaling of the susceptibility (and cluster sizes) despite vortices being allowed. This "QLRO by disorder" regime only exists for sufficiently large $p'$, a statement that will be made quantitative shortly.

\subsubsection{Second moment correlation length}
We now take Eq.~\ref{eq:etaKT} as our \textit{definition} of the critical angle, i.e., $\Delta=\Delta_c$ is the angle at which $\eta=1/4$. If the transition is of a Kosterlitz-Thouless type, then this definition of $\Delta_c$ should be consistent with that estimated via other metrics. One of these is the second moment correlation length, $\xi_{2nd}(L)$, is defined as~\cite{Hasenbusch2005}:
\begin{equation}
\xi_{2nd}(L) = \frac{1}{2\sin(\pi/L)}\sqrt{G(0,0)/G(2\pi/L,0) - 1}.
\end{equation}
This quantity takes its name from the fact that it is the second moment with respect to the Fourier-transformed correlation function, $G(\vec{k})$. In the paramagnet, $\xi_{2nd}(L)$ is independent of system size, whereas in the critical phase, it scales linearly with system size $L$. Thus, the rescaled correlation length, $\xi_{2nd}(L)/L$, is independent of $L$ up to the critical angle, up to subleading corrections at small system sizes. In the thermodynamic limit at the KT transition, this rescaled length takes the value~\cite{Hasenbusch2005}:
\begin{equation}\label{eq:xi2-KT}
    \lim_{L\rightarrow\infty} \frac{1}{L} \xi_{2nd}(L, T = T_{KT}) = 0.7506912.
\end{equation}
For the XY model, $\xi_{2nd}(L)/L$ was found to quickly approach this limit from below~\cite{Komura2012}. 

The hypothesis that our class of models exhibits a KT transition thus yields two predictions: (1) $\xi_{2nd}/L$ should be independent of system size up to $\Delta_c$, and (2) it should take the value given by Eq.~\ref{eq:xi2-KT} at $\Delta=\Delta_c$. (At the system sizes reported in this paper we do not expect this limit to be obtained with great accuracy, but it should fall in the rough neighborhood.) These predictions both approximately hold, as shown in the left panel of Fig.~\ref{fig:KT}. Deviations from $\xi_{2nd}/L\approx$ constant appear above $\Delta\approx 0.57\pi$, and $\xi_{2nd}/L=0.7506912$ between $0.56\pi$ and $0.57\pi$, consistent with the estimate of $\Delta_c$ obtained from Eq.~\ref{eq:etaKT}.

    \begin{figure}[t]
    \centering
    \includegraphics[width=\linewidth]{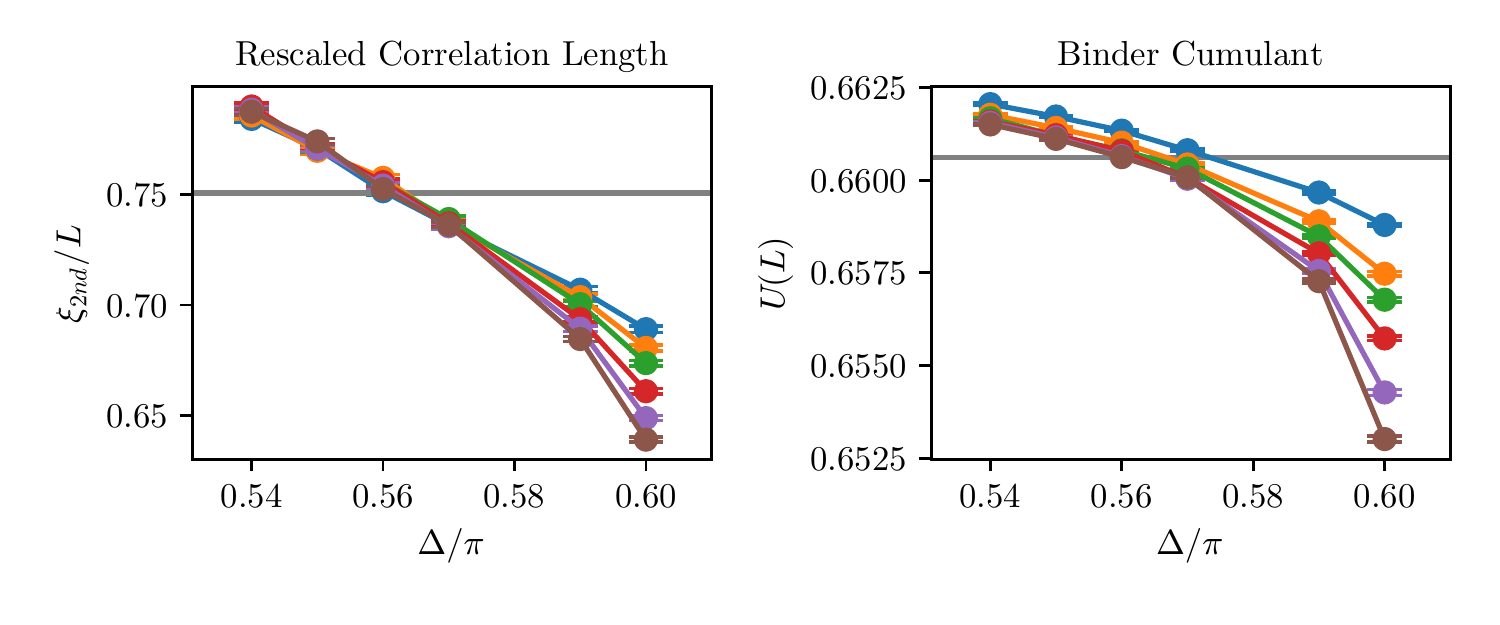}
    \caption{Data collapse of dimensionless quantities, the rescaled second moment correlation length and Binder cumulant, near the KT phase transition. System sizes from top (blue) to bottom (brown) are $L=8$, 16, 24, 40, 60, and 80. The gray line in the left panel marks the value of $\lim_{L\rightarrow\infty}\xi_{2nd}/L=0.7506912$ at vortex unbinding~\cite{Hasenbusch2005}; gray line in the right panel marks the value of the Binder cumulant $U=0.660603$ at the KT transition~\cite{Hasenbusch2008}.\label{fig:KT}}
    \end{figure}
    
\subsubsection{Binder cumulant}
The Binder cumulant $U(L)$, defined as:
\begin{equation}
U(L) = 1 - \langle |\vec{M}|^4 \rangle/3\langle |\vec{M}|^2\rangle ^2
\end{equation}
is, like the rescaled correlation length, asymptotically independent of $L$ within the critical phase. Therefore, for a system which undergoes a KT transition we should observe data collapse between different system sizes up to $\Delta=\Delta_c$, above which $U(L)$ will decay with $L$. At the KT transition point, $U$ takes the value~\cite{Hasenbusch2008}:
\begin{equation}\label{eq:binder-KT}
    \lim_{L\rightarrow\infty} U(L, T=T_{KT}) = 0.660603.
\end{equation}
In the XY model, the approach to this limit is sufficiently rapid ($U(L=32)$ is within 0.02\% of the limiting value) that we expect Eq.~\ref{eq:binder-KT} to be well approximated for the finite system sizes simulated in this study. Our general expectations are confirmed, as shown in the right panel of Fig.~\ref{fig:KT}. Again, $U(L)$ is roughly independent of $L$ within uncertainties up to $\Delta \approx 0.57\pi$, crossing the value given in Eq.~\ref{eq:binder-KT} at $\Delta \approx 0.56\pi$. Interestingly, {as was pointed out in an earlier study of the constraint-only model~\cite{Bietenholz2013a}}, the $L\rightarrow\infty$ limit is approached from above, consistent with the $1/\log L$ scaling corrections of a Gaussian model, but of opposite sign to the scaling corrections measured for the members of the universality class studied in Ref.~\cite{Hasenbusch2008}, including the nearest-neighbor XY and Villain models. In that work, a large $1/(\log L)^2$ correction was attributed to the presence of vortices neglected in the spin-wave theory. While it is not clear where our highly nonlinear model falls on the RG flow to the Gaussian theory, we suspect that the strong suppression of defects which persists even above $\pi/2$ could explain the absence of this vortex-driven correction.

In summary, the various measures discussed in this section---qualitative estimate of the goodness of fit to the susceptibility, the value of $\eta$ determined from this fit, the system size independence of $\xi_{2nd}/L$ and $U(L)$, and their agreement with Eq.~\ref{eq:xi2-KT} and Eq.~\ref{eq:binder-KT} in the vicinity of the transition---collectively lend support to the hypothesis that, in the XY limit, our new class of inclusion models undergoes a Kosterlitz-Thouless transition between $\Delta=0.56\pi$ and $\Delta=0.57\pi$. 

\subsection{Scaling of $\eta$ and QLRO by Disorder}
Recall that the Kosterlitz-Thouless phase is characterized by a line of critical points, with continuously varying exponents. While in the standard case $\eta$ is a function of temperature, in the hardcore spin models $\eta$ is a function of $\Delta$ instead. In this section we find that $\eta(\Delta)$ takes a simple functional form, and relate it to the subtler features of the phase diagram in Fig.~\ref{fig:phase}.

In Fig.~\ref{fig:etascaling} we perform a power-law fit to $\eta(\Delta)$ for inclusion angles ranging from $\Delta=0.05\pi$ to $\Delta=0.56\pi$. Strikingly, a single trend
\begin{equation}\label{eq:scaling}
\eta(\Delta) \sim \Delta^2
\end{equation}
holds for the entire range of $\Delta$. In particular, the functional form of $\eta(\Delta)$ does not appear to change as the system crosses over from the defect-forbidden region ($\Delta < 0.5\pi$) to the "QLRO by disorder" region marked in Fig.~\ref{fig:phaseline}. This is a manifestation of universality in that the critical exponent is not sensitive to whether vortices are strictly forbidden by the hard constraint. This perhaps makes the simplicity of Eq.~\ref{eq:scaling} all the more surprising, as it relates the scaling exponent of a macroscopic observable and the microscopic parameter $\Delta$ in such a straightforward way.
    \begin{figure}[t]
    \centering
    \subfloat[]{
    \centering
    \includegraphics[width=\linewidth]{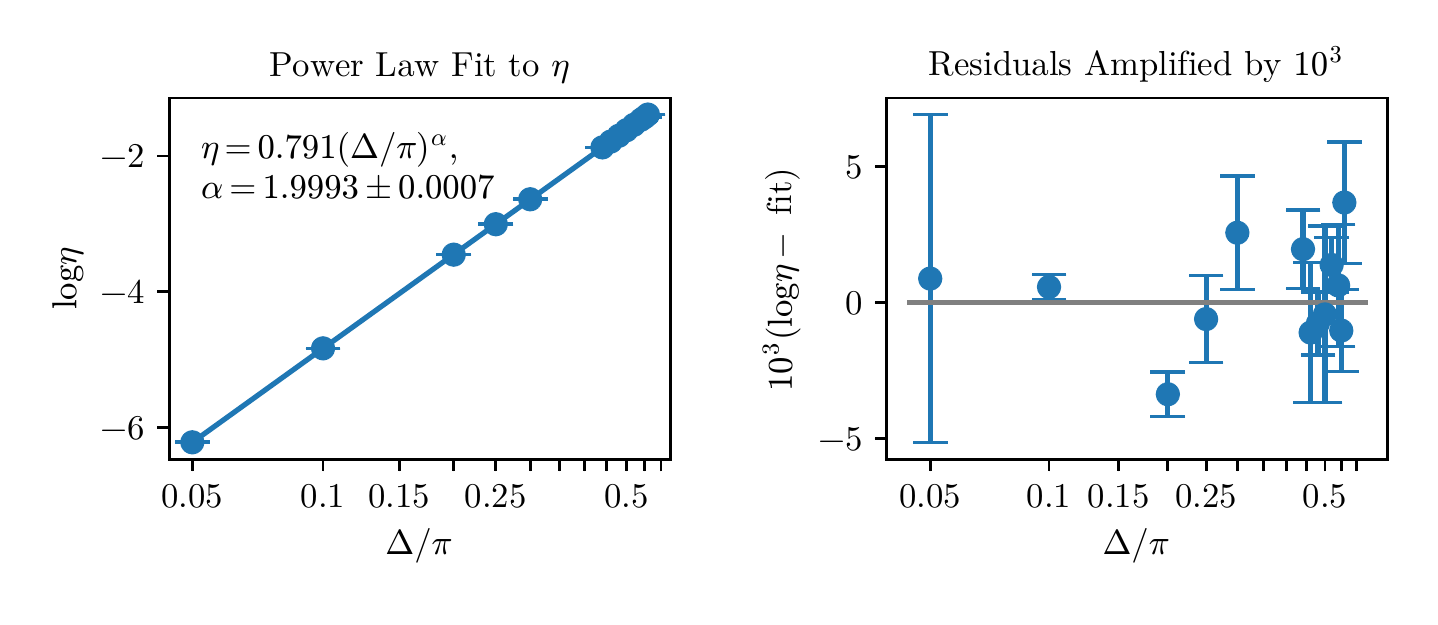}
    \label{fig:etascaling}
    }
    \\
    \subfloat[]{
        \centering
    \includegraphics[width=\linewidth]{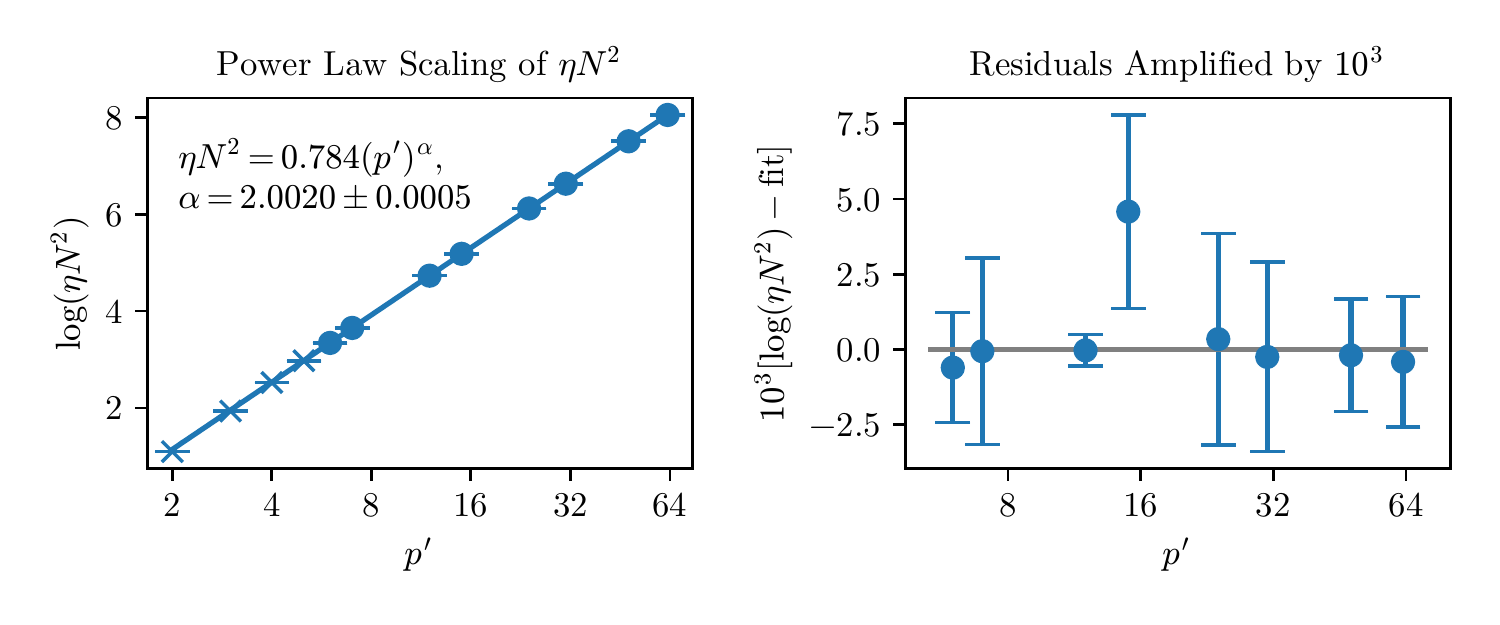}\label{fig:stiffness}
    }
    \caption{Scaling of the critical exponent $\eta$ in the critical phase. (a) Power-law scaling as a function of $\Delta$ in the XY limit. Error bars are derived from the statistical uncertainties in the power-law fits to the susceptibility. The right panel shows residuals with respect to the fit $\eta=0.79\Delta^{1.999}$, amplified by a factor of $10^3$. The $\Delta/\pi$ axis uses a logarithmic scale. (b) Scaling of $\eta N^2$ as a function of $p'$ for $p'$ families of $Z_N$ models. Error bars are statistical uncertainties from power-law fits to the susceptibility of one model within the $p'$ family (except $p'=2$ where the theoretical value of 3 is used). The $p'$ axis uses a logarithmic scale. The right panel shows the residuals with respect to the fit $\eta N^2 = 0.784p'^2$, starting at $p'=6$.}
    \end{figure}
Assuming that Eq.~\ref{eq:scaling} holds yields a more precise estimate of the critical angle $\Delta_c$, by solving for the angle at which $\eta=1/4$. A single parameter fit fixes the constant of proportionality and implies
\begin{equation}\label{eq:delta-c}
\Delta_c = (0.56192 \pm 0.00014)\pi.
\end{equation}
It must be emphasized that the error bars here only capture the statistical uncertainty of the fit to Eq.~\ref{eq:scaling} and not any of the uncertainties associated with the limited number of independent samples, logarithmic corrections to scaling for small system sizes, and so on. {In particular, taking logarithmic corrections into account in the fit to susceptibility for system sizes sizes up to $L=4096$ yielded a value for $\Delta_c$ that is $0.5\%$ greater than our estimate~\cite{Bietenholz2013a}. Nevertheless, it is remarkable that we could come this close to the precision determination of $\Delta_c$ from such a simple ansatz for $\eta(\Delta)$.}

Returning to clock models, in Fig.~\ref{fig:stiffness} we fit $\eta N^2$, the invariant quantity within $p'$ families, to a power law as a function of $p'$ and find:
\begin{equation}\label{eq:pp2}
    \eta N^2 (p') = 0.784 (p')^\alpha, \qquad \alpha = 2.0020 \pm 0.0005
\end{equation}
for a fit beginning at $p'=6$.
This is equivalent to Eq.~\ref{eq:scaling} once we make the identification $\Delta=\pi p' /N$. This equivalence was partly anticipated by Fig.~\ref{fig:KTeta}, in which clock models, both inside and outside their respective $p'$ families, followed a similar trend in $\eta$ as models in the XY limit near an inclusion angle of $0.5\pi$. But given the presence of the symmetry-breaking $h_N$ interaction which forces the spins to be discretized, it is not obvious beforehand that models of $Z_N$ and XY spins will exhibit the same scaling once we have mapped the parameters $(N,p')\leftrightarrow \Delta$. Thus, a few explanatory remarks are in order.

First, the monotonic increase of $\eta N^2$ as a function of $p'$ implies the absence of a symmetry-broken (smooth) phase in the phase diagram in Fig.~\ref{fig:phase}. (We exclude $p'=1$, which is trivially ferromagnetic or antiferromagnetic and maps to an interface of constant height.) If there existed a region of $(N,p')$ parameter space that possessed long-ranged order, it would overlap with the defect-forbidden region, and thus we could identify some $p'$ family with $\eta_{lock} < 4$. For odd $p'$, this would imply (Eq.~\ref{eq:lock}) $\eta N^2 < 4$, while for even $p'$, this would imply (Eq.~\ref{eq:lock2}) $\eta N^2 < 1$. Since small $p'$ families are in the rough phase, as determined both from finite size scaling of the transfer matrix and from the \texttt{reflect} algorithm, the monotonicity of $\eta N^2 (p')$ implies that \textit{all} $p'$ families are in the rough phase. This is a peculiar result, as nothing in our class of models forbids a smooth phase \textit{a priori}. But the fact that the locking potential becomes strongly irrelevant as $p'$ increases is also internally consistent with the observation that, for $p'$ of order 10 and above, there is no real distinction in terms of critical exponents between $Z_N$ models and XY models of commensurate inclusion angle. {This is reminiscent a known feature of the phase diagram of $Z_N$ models with standard action: while the $Z_5$ model exhibits a lower transition temperature between the QLRO and paramagnetic phase as well as stronger finite-size effects in the partition function zeros and helicity modulus, these quantities are roughly independent of $N$ for $N\geq 6$, collapsing onto the XY limit~\cite{Kim2017,Li2020}.}

 Having ruled out a smooth phase, there remain two other potential phases for clock models: the critical phase and the paramagnetic phase. As indicated by the unshaded region in Fig.~\ref{fig:phase} for a fixed $N\geq 4$ there exists a sufficiently large $p'<N$ such that the corresponding $(N,p')$ model is nontrivially paramagnetic. This phase includes, for example, the $q$-state Potts antiferromagnets, which are disordered for $q\geq 4$ on the square lattice~\cite{Salas1998}. This leaves open the question of whether all models outside the $p'$ families are disordered, or whether there exists a regime, as in the XY limit, in which defects are allowed but irrelevant. 

For small $p'$, there is no such regime, a feature that can be understood via the height representation. A defect in the height field of winding number 1 has $\Delta h = N$, which, when substituted into the effective free energy in Eq.~\ref{eq:effective}, implies that the most relevant vortex operator has critical exponent~\cite{Cardy}:
\begin{equation}
    \eta_{vortex} = \frac{N^2 K}{2\pi} = \frac{1}{\eta}.
\end{equation}
Within a $p'$ family, if $\eta>1/4$ for some $N$, this implies that the vortex operator is relevant, and quasi-long-ranged order is preserved only because the strength of the operator is identically zero. In that case, the transition from QLRO to the paramagnet as a function of $N$ for fixed $p'$ is not Kosterlitz-Thouless in nature but geometric, occurring exactly at $N=N_c(p')$ when defects become allowed. This is the case, for example, for $p'=2$, as the smallest member of the family, the 3-state Potts AFM, has $\eta=1/3$. Within the range of $p'$ exhibiting this geometric transition, the corresponding value of $\eta N^2$ also deviates slightly from the trend in Eq.~\ref{eq:pp2}. Including the data points with $p'<6$ leads to a greater reduced $\chi^2$ value and a noticeable trend in the residuals indicative of deviations from a power law, although it does not significantly change the predicted exponent. These deviations from the trend possibly reflect the greater influence of the locking potential at low $p'$, as well as the fundamental contrast between the geometric nature of the transition for low $p'$ and the KT transition for XY spins.

As $p'$ increases, however, the effective inclusion angle of the smallest member of the $p'$ family decreases, and the associated critical exponent decreases. Evaluating Eq.~\ref{eq:pp2} at $N=N_c(p')+1$ indicates that for $p'\geq 5$, all members of the corresponding $p'$ family have $\eta<1/4$. This allows for critical models outside the family. The minimum value of $p'$ for which this QLRO by disorder phase exists can be estimated by positing that the critical exponent in the vortex-allowed critical region obeys the same scaling form as in the vortex-forbidden region, and evaluating Eq.~\ref{eq:pp2} at $N=N_c(p')$:
\begin{equation}
    \eta(N_c, p') = 0.784(p')^2/(2p'-2)^2
\end{equation}
which implies a crossover at $p'\approx 8.7$. Since $p'$ only takes integer values, this indicates that the $(N,p')=(16,9)$ model is the smallest $p'$ critical model not described by a height representation. 

\section{Conclusions}\label{sect:conclude}
It is a marvel of statistical mechanics that systems with drastically different microscopic physics can exhibit the same macroscopic behavior near critical points. The class of "constraint-only" spin models examined in this paper serve as an extreme example of this phenomenon. The divergent nature of the potential means that temperature is not well-defined, and the ordering and phase transitions are driven solely by entropy. It is not immediately obvious that these phenomena, examined here for models of $Z_N$ and XY spins on the square lattice, will fall into the same universality classes as their finite-temperature cousins and yet, when the dust settles, they do. That said, there are many interesting details, most notably the existence of a vortex free and yet non-linear ``spin wave'' regime and that the long-range ordered phases that {\it can} exist for $Z_N$ models do not arise in our models. We have investigated our models using height representations, transfer matrix calculations, and an ergodic cluster algorithm, arriving at the phase diagram shown in Fig.~\ref{fig:phase}.

Our focus in this paper has been on the equilibrium properties of this class of spin models. Yet the simplified nature of the interaction, constructed by analogy to that of hard spheres/disks, also imbues the spin models with interesting dynamics.  This dynamics will be addressed in a following paper. Indeed as we were finishing this work, we came across Ref.~\cite{Yoshino2018} which has applied ``replica symmetry breaking'' techniques to related models on random graphs, and it will also be instructive to make contact with this line of work.  

\section*{Acknowledgments}
We would like to thank Sal Torquato, Frank Stillinger, Michael Aizenman, Ron Peled, and Romain Vasseur for useful discussions. We also would like to thank an anonymous referee for bringing to our attention earlier works from lattice field theory that cover similar ground, particularly Ref.~\cite{Bietenholz2013a}. GMS is supported by the Department of Defense (DoD) through the National Defense Science \& Engineering Graduate (NDSEG) Fellowship Program and SLS would like to acknowledge the support of the Department  of  Energy  via  grant  No.   DE-SC0016244. Additional support was provided by the Gordon and Betty Moore Foundation through Grant GBMF8685 towards the Princeton theory program and the use of the TIGRESS High Performance Computer Center at Princeton University.
This work was in part supported by the Deutsche Forschungsgemeinschaft under grants SFB 1143 (project-id 247310070) and the cluster of excellence ct.qmat (EXC 2147, project-id 390858490). 
\appendix
\section{Ergodicity of the \texttt{Reflect} Algorithm}\label{sect:app1}
In this section we prove that Algorithm~\ref{alg:cluster} is ergodic for inclusion models on any lattice. Furthermore, ergodicity of exclusion models on bipartite lattices is proven. While this has already been proven elsewhere~\cite{Hasenbusch1996}, the alternate proof provided here will also serve to introduce the notation employed in the following Appendix, for a closer study of the algorithm's properties.

Since each move is reversible, in proving ergodicity it suffices to prove that a chosen target configuration can be attained from any initial state via a finite sequence of cluster moves~\cite{Burton1997, Moore2000}.
\begin{theorem}\label{theor:inclusion}
Let $T$ be a ferromagnetic target configuration on a lattice of $N_s$ spins, with $\theta=0$ for all spins. Given any initial configuration $A$ which respects the inclusion constraint, $A$ can be transformed to $T$ by a series of $\leq N_s$ iterations of the \texttt{reflect} algorithm.
\end{theorem}
\begin{proof}
Let $S$ denote the (possibly empty) domain of spins on which $A$ agrees with $T$. To transform $A$ into $T$, we successively add spins to this domain, with each move adding $\geq 1$ spin to the domain and leaving the spins in $S$ unchanged.

A cluster move starts by adding a random spin $s$ to the pocket $\mathcal P$. Without loss of generality, let $s=\exp(2i\theta)$, where $0<\theta<\pi$. A reflection about the axis $\theta$ will align $s$ with the ferromagnetic domain, that is, it will add $s$ to $S$. If $C$ denotes the set of all spins that participate in the cluster move, i.e., all spins that are at some point added to the pocket $\mathcal P$, then it suffices to prove that $C \cap S = \emptyset$. 

For convenience, we redefine our angles with respect to the axis $\theta$, so that $s=\exp(i\theta)$, and all spins in $S$ are aligned with $-\theta$. For reasons that will become clear in the next section, each spin on the lattice can be written in the form $s_i = \exp(i \sigma_i \phi_i)$, where $\sigma_i = \pm 1$ is an Ising variable, and  $\phi_i \in [0, \pi]$. Then, the proposed reflection flips the Ising variable while leaving $\phi$ unchanged. 

A spin $s_j$ is added to the pocket if and only if one of its neighbors $s_i$ is reflected and if, after the reflection, the pair violates the inclusion constraint. This leads to the following proposition:
\begin{prop}\label{prop:same}
Consider two neighboring spins $s_i=\exp(i\sigma_i \phi_i)$, $s_j = \exp(i\sigma_j \phi_j)$, where $\phi \in [0, \pi]$, $\sigma = \pm 1$. If, after reflecting $s_i$ about the axis $\theta=0$, i.e. $s_i\rightarrow s_i^*$, the pair $(s_i^*, s_j)$ violates the inclusion constraint, then $\sigma_i = \sigma_j$.
\end{prop} 
For, since each move starts from a valid configuration, $s_i$ and $s_j$ must initially enclose an angle less than $\Delta$:
\begin{equation}\label{eq:cosd1}
\cos\Delta < \cos(\sigma_j \phi_j - \sigma_i\phi_i) = \cos\phi_i\cos\phi_j + \sigma_i \sigma_j \sin\phi_i\sin\phi_j
\end{equation}
while, after reflecting the spin $s_i$, the inclusion constraint will only be violated if the inner product satisfies:
\begin{equation}\label{eq:cosd2}
\cos\Delta > \cos(\sigma_j \phi_j + \sigma_i\phi_i) =
\cos\phi_i\cos\phi_j - \sigma_i \sigma_j \sin\phi_i\sin\phi_j.
\end{equation}
Inequalities~\ref{eq:cosd1} and~\ref{eq:cosd2} can only hold simultaneously if $\sigma_i = \sigma_j$.

Every spin in the set $C$ can be viewed as the endpoint of a directed path of spins  $s_1 \rightarrow s_2 ...\rightarrow s_n$ where $s_{i}$ gets reflected, which then violates the hard constraint with its neighbor $s_{i+1}$, which is then added to the pocket, and so on. Since the cluster move starts by transforming $s$, every such path can be traced back to $s$, and since all spins on the path have the same Ising variable due to Proposition~\ref{prop:same}, it follows that every spin in $C$ has $\sigma=1$. Recalling our convention that all spins in the ferromagnetic domain $S$ have a negative Ising variable, Theorem~\ref{theor:inclusion} immediately follows.
\end{proof}
The same proof applies for inclusion $Z_N$ models, the only difference being that the set of initial allowed configurations is restricted to discretized spin states. 

Theorem~\ref{theor:inclusion} is independent of the choice of lattice. On bipartite lattices, the algorithm is also ergodic for exclusion models, including those which lack a direct mapping to inclusion models, namely, odd $N$ clock models. In this case, the target configuration consists of $\theta = 0$ on the A sublattice and $\theta=\pi (N-1)/N$ on the B sublattice, approaching an antiferromagnet in the XY limit. Using the same definitions as before, with angles defined with respect to the axis of reflection, we arrive at the following corollary to Prop.~\ref{prop:same}:
\begin{prop}\label{prop:opp}
Consider two neighboring spins $s_i=\exp(i\sigma_i \phi_i)$, $s_j = \exp(i\sigma_j \phi_j)$, where $\phi \in [0, \pi]$, $\sigma = \pm 1$. If, after reflecting $s_i$ about the axis $\theta=0$, i.e. $s_i\rightarrow s_i^*$, the pair $(s_i^*, s_j)$ violates the exclusion constraint, then $\sigma_i = - \sigma_j$.
\end{prop} 
This implies that $C$ contains only $\sigma=1$ ($\sigma=-1$) on the A (B) sublattice, whereas $S$ contains only $\sigma=-1$ ($\sigma=1$) on the A (B) sublattice. Once again, $C \cap S = \emptyset$, so the target configuration can be attained in $\leq N_s$ moves.
\section{Mapping to Random Cluster Model}\label{sect:app2}
    \begin{figure}[hbtp]
    \centering
    \includegraphics[width=\linewidth]{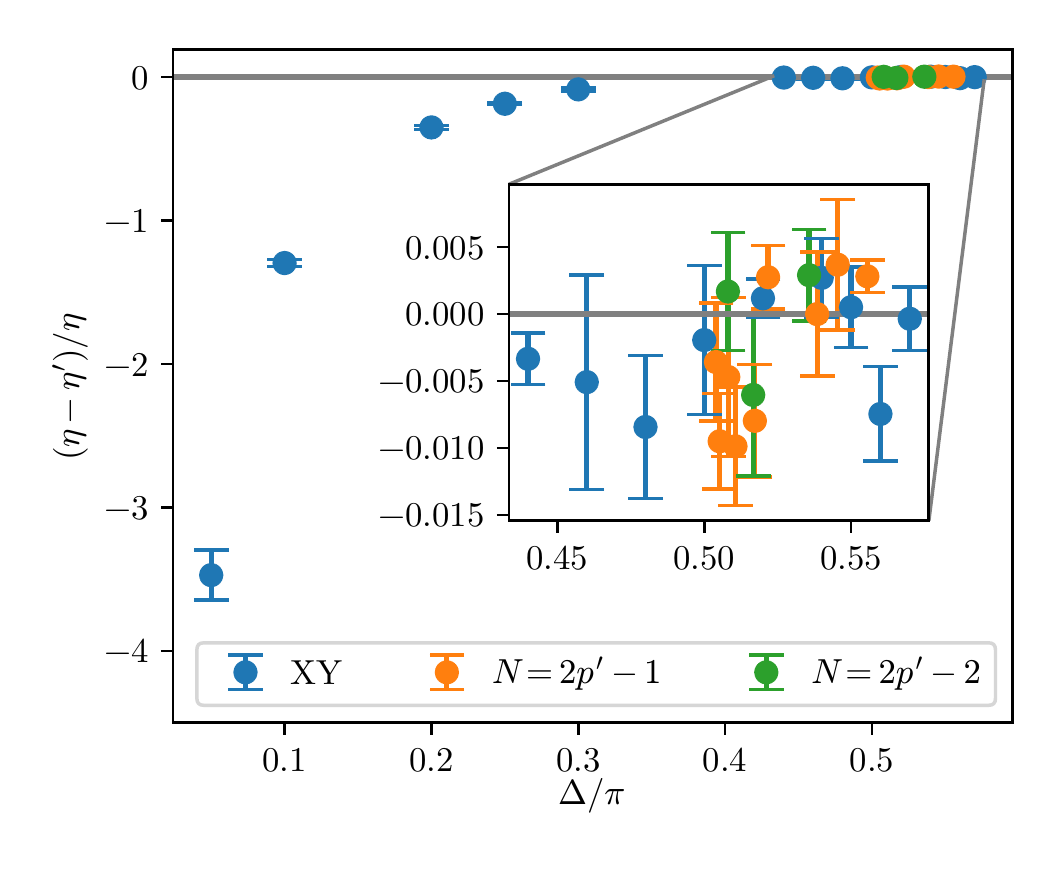}
    \caption{Discrepancy between $\eta$ and $\eta'$ as a function of inclusion angle. Vertical axis is the relative error $(\eta-\eta')/\eta$, with error bars estimated from the statistical uncertainty of the power law fit used in finite size scaling. Data from clock models were plotted by defining $\Delta = \pi p'/N$.\label{fig:eta-v-etap}} 
    \end{figure}
\subsection{Fortuin-Kasteleyn Mapping}
An appealing feature of the \texttt{reflect} algorithm is that the average cluster size near the transition scales with the same exponent, within error bars, as the susceptibility, as shown in Fig.~\ref{fig:eta-v-etap}. This property is observed for "good" cluster algorithms, such as the Wolff algorithm for Ising and Potts models~\cite{Wolff1989} and a generalized geometric cluster algorithm used to study lattice gases and Ising models at constant magnetization~\cite{Heringa1998,Heringa1998a}. Conversely, the pivot cluster algorithm is inefficient for systems of hard disks, and even more so for hard rods, near the transition because the percolation threshold and transition density do not match~\cite{Krauth2003}. Thus, agreement between $\eta$ and $\eta'$ is a measure of the algorithm's efficiency, indicating the extent to which the clusters produced by the algorithm are the "physical" clusters mediating the transition~\cite{DOnorioDeMeo1990}.

The scaling of cluster sizes in the Wolff algorithm follows from an exact mapping between the Ising Model and the Fortuin-Kasteleyn random cluster model~\cite{Fortuin1972a,Fortuin1972b, Fortuin1972}. This equivalence is constructed from a joint probability distribution over Ising spins $\{\sigma_i\}$, which live on the sites of the lattice, and bond variables $\{b_{ij}\}$ defined on the nearest-neighbor links. Each bond is either occupied ($b_{ij}=1$) or empty ($b_{ij}=0$). A given bond configuration decomposes the lattice into clusters, where two sites $(i, j)$ are said to be in the same cluster if and only if they are connected by a path of occupied bonds, denoted $i\leftrightarrow j$. The Swendsen-Wang algorithm uses this bond configuration to generate the next spin state by independently flipping each cluster with probability $1/2$~\cite{Swendsen1987}. On the other hand, the Wolff algorithm selects just one of these clusters and flips it with probability 1; this single cluster is constructed by choosing a random site and adding neighbors to the cluster with the conditional probability $p(b_{ij}=1|\{\sigma\})$. This leads to two questions: (1) Can the scaling of clusters in the \texttt{reflect} algorithm likewise be explained through an exact mapping to a bond model? (2) If such a mapping exists, what drives the systematic discrepancy between $\eta$ and $\eta'$ at lower inclusion angles seen in Fig.~\ref{fig:eta-v-etap}?

As a first step toward such a mapping, we define an Ising variable ("pseudospin") on each site and relate the pseudospin correlation function to the correlation function of the original spins, following a similar procedure to that in Refs.~\cite{Patrascioiu1993,Aizenman1994,Patrascioiu2002}. As in Appendix~\ref{sect:app1}, we let $\sigma=\mathrm{sgn}(\varphi)$, where $\varphi\in [-\pi, \pi]$ is the angle with respect to the axis of reflection $\theta$. Then, expressing the angle $\varphi = \sigma \phi$ where the "auxiliary spin" $\phi$ is in the interval $[0,\pi]$, a reflection about axis $\theta$ flips the Ising variable while leaving the auxiliary spin unchanged. The inclusion models considered in this paper belong to a class of ferromagnetic Hamiltonians whose pair potential $V(\varphi_i, \varphi_j)$ obeys the inequality:
\begin{equation}
    V(\phi_i, \phi_j) = V(-\phi_j, -\phi_j) \leq V(\phi_i, -\phi_j) = V(-\phi_i, \phi_j).
\end{equation}
That is, given a choice of auxiliary spins $\phi_i, \phi_j$, the pair potential between like Ising pseudospins is less than or equal to the pair potential between opposite Ising pseudospins. This suggests the following definition of the bond occupation probability, conditioned on both $\{\sigma\}$ and $\{\phi\}$:
\begin{equation}\label{eq:bij}
    p(b_{ij}=1|\{\sigma\}, \{\phi\}) = \delta_{\sigma_i\sigma_j} p_{ij}(\{\phi\})
\end{equation}
where
\begin{equation}\label{eq:phi-pij}
    p_{ij}(\{\phi\}) = 1 - \exp[V(\phi_i, \phi_j) - V(-\phi_i, \phi_j)].
\end{equation}
Note that $p_{ij}$ depends only on the auxiliary spins; the conditional bond occupation probability in Eq.~\ref{eq:bij} depends on the Ising pseudospin only through the Kronecker delta function. This allows us to define the joint probability distribution of the $\{\sigma\}, \{b\}$ variables, conditioned on a given configuration of auxiliary spins, as:
\begin{equation}
    p(\{\sigma\},\{b\} | \{\phi\}) = Z^{-1} \prod_{\langle ij \rangle}[(1-p_{ij})\delta_{b_{ij}, 0} + p_{ij}\delta_{b_{ij}, 1}\delta_{\sigma_i\sigma_j}
\end{equation}
where $Z$ is the partition function used to normalize the probability distribution. Summing out over the bond variables yields the marginal distribution for $\{\sigma\}$ conditioned on $\{\phi\}$:
\begin{equation}\label{eq:p-sigma}
    p(\{\sigma\}|\{\phi\})=  Z^{-1} \prod_{\langle ij \rangle} e^{V(\phi_i, \phi_j)} \prod_{\langle ij \rangle}e^{-V(\sigma_i\phi_i, \sigma_j\phi_j)}.
\end{equation}
The first product depends only on the auxiliary variables while the second product is simply the Boltzmann weight associated with a given spin configuration; supplemented with an appropriately normalized definition of the $\{\phi\}$ probability distribution as $p(\{\phi\}) \propto \prod_{\langle ij \rangle}\exp(-V(\phi_i, \phi_j))$, we recover the partition function of the original Hamiltonian. Eq.~\ref{eq:p-sigma} further implies the conditional bond occupation probability as defined in Eq.~\ref{eq:bij}. The marginal distribution for the bond variables also has a simple form~\cite{Edwards1988}, but in the present context we are most interested in the resulting form of the pseudospin probability distribution conditioned on the bond variables. As with the Fortuin-Kasteleyn representation of the Ising model, Ising spins belonging to the same cluster are aligned, while different clusters are uncorrelated. Summing over all possible bond configurations, the expectation value of the Ising variable correlation is:
\begin{align}\label{eq:ising}
    \langle \sigma_i \sigma_j \rangle 
    = p(i \leftrightarrow j)
\end{align}
where 
$p(i \leftrightarrow j)$ is the probability, over all possible bond configurations, that $i$ and $j$ belong to the same cluster. Eq.~\ref{eq:ising} is the central identity to the FK mapping. 

While this construction works for any ferromagnetic Hamiltonian defined on XY spins (and can straightforwardly be generalized to antiferromagnetic Hamiltonians by considering a staggered correlation function), it is most useful in the context of our hard-core inclusion models. In this case, the bond configuration is fully determined by the spin configuration:
\begin{equation}
    p_{ij}(\phi) = \begin{cases} 1 & \mathrm{angle}(-\phi_i, \phi_j) \geq \Delta \\
    0 & \mathrm{angle}(-\phi_i, \phi_j) < \Delta
    \end{cases}.
\end{equation}
In words, this equation says that the bond $\langle ij \rangle$ is occupied if and only if reflecting spin $s_i$ about the chosen axis would cause the inclusion constraint with $s_j$ to be violated. This is precisely the condition for adding a spin to the cluster in line 8 of Algorithm~\ref{alg:cluster}. Thus, the \texttt{reflect} algorithm identifies and transforms one of the clusters in the bond configuration $\{b\}$, which in turn is determined by the original spin state and the randomly chosen axis of reflection.

\begin{figure}[t]
\includegraphics[width=0.6\linewidth]{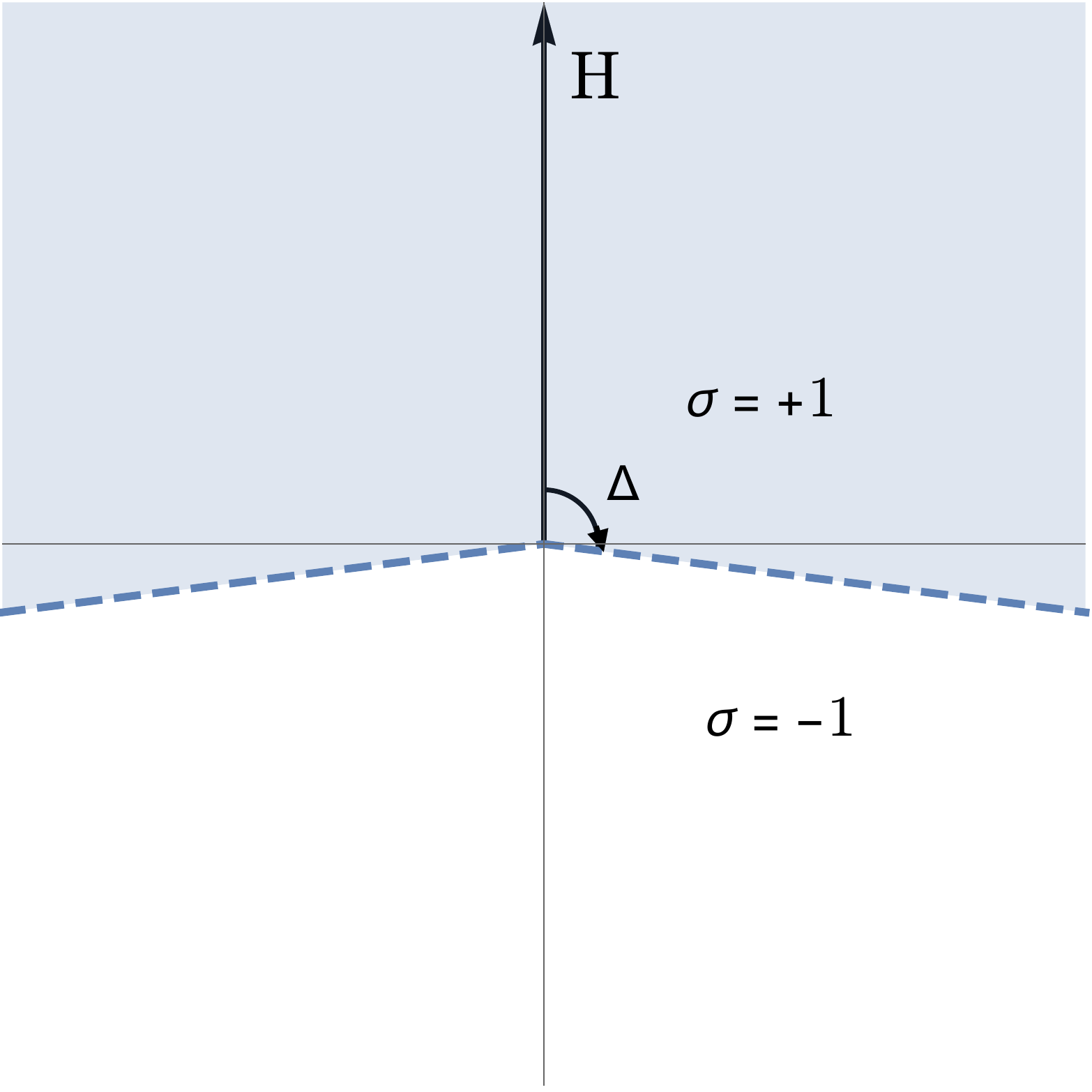}
\caption{\label{fig:ctheta} Sketch of the relationship between the Ising variable $\sigma$ and the piecewise function $C(\vartheta)$ defined in Eq.~\ref{eq:Ctheta}. For the Ising pseudospins, $\sigma=1$ for spins pointed above the $x$ axis, and $\sigma=-1$ for spins pointed below the $x$ axis. For the function $C(\vartheta)$, we define $\vartheta$ with respect to an external magnetic field $\vec{H}$ oriented along the positive $y$ axis. Then, given an inclusion angle $\Delta$, spins pointing in the blue shaded region have $C(\vartheta)=1$ and spins pointing in the unshaded region have $C(\vartheta)=-1$. For the inclusion angle shown, $C(\vartheta)=\sigma$ everywhere except in the narrow shaded region below the $x$ axis.}
\end{figure}

To relate Eq.~\ref{eq:ising} to the XY spin correlation function $\langle \vec{s}_i \cdot \vec{s}_j\rangle$, consider instead the correlation function $\langle C(\vartheta_i) C(\vartheta_j)) \rangle$, where:
\begin{equation}\label{eq:Ctheta}
    C(\vartheta) = \begin{cases} 1 & |\vartheta| < \Delta \\
    -1 & |\vartheta| \geq \Delta
    \end{cases}
\end{equation}
with $\vartheta$ the angle with respect to an external field $\vec{H}$. This correlation function was used to study the phase diagram of a finite-temperature step model with $\Delta=\pi/2$ in Ref.~\cite{Guttmann1973,Guttman1978}, in which it was speculated that the associated susceptibility scales with the same exponent as that associated with $\langle \vec{s}_i \cdot \vec{s}_j \rangle$. To justify this speculation, note that $C(\vartheta)$ can be expanded as a cosine series on the interval $[-\pi, \pi]$, yielding the correlation function:
\begin{equation}\label{eq:Cthetacorr}
\langle C(\vartheta_i) C(\vartheta_j)) \rangle = \sum_{n=1}^{\infty} a_n^2 \langle \cos(n\vartheta_i)\cos(n\vartheta_j)\rangle
\end{equation}
where $a_n$ is the Fourier coefficient of the $n$th term in the series and charge neutrality ensures that cross terms vanish~\cite{Fradkin2013}. Since the associated critical exponent of each term in the series scales as $n^2$, Eq.~\ref{eq:Cthetacorr} is dominated by the first term, which is proportional to $\langle\vec{s}_i \cdot \vec{s}_j\rangle$. The justification for using the $C(\vartheta)$ correlation function as a proxy for the XY spin correlation function is especially strong at an inclusion angle of $\Delta=\pi/2$, for which the coefficients $a_n$ are identically zero for even $n$. In this case, the subleading correction is $\langle \cos(3\vartheta_i) \cos(3\vartheta_j) \rangle$ which, although relevant, is strongly suppressed compared to the leading term $\langle \cos\vartheta_i \cos\vartheta_j \rangle$.

At $\Delta=\pi/2$, $C(\vartheta)$ is also in exact correspondence with the Ising variable $\sigma$, if we align $\vec{H}$ at an angle of $\pi/2$ with respect to the axis of reflection (Fig.~\ref{fig:ctheta}). Eq.~\ref{eq:ising} then implies:
\begin{equation}\label{eq:cthetacorrb}
    \langle C(\vartheta_i) C(\vartheta_j) \rangle =  p(i \leftrightarrow j) \qquad (\Delta=\pi/2).
\end{equation}
For inclusion angles near $\pi/2$, such as those in the neighborhood of the Kosterlitz-Thouless transition, Eq.~\ref{eq:cthetacorrb} approximately holds. Then, summing over lattice sites $i, j$ to obtain the susceptibility and recalling that the \texttt{reflect} algorithm selects one cluster in the random bond configuration with probability proportional to the size of the cluster, we arrive at the final result:
\begin{equation}\label{eq:susc}
    L^{2-\eta} \sim \chi \sim \langle s \rangle \sim L^{2-\eta'}.
\end{equation}
Eq.~\ref{eq:susc} justifies the close agreement of $\eta$ and $\eta'$ near the critical angle $\Delta_c\approx 0.56\pi$, which is fortuitously close to the angle at which Eq.~\ref{eq:cthetacorrb} holds exactly. On the other hand, the sharp deviation of $\eta'$ from $\eta$ for low inclusion angles seen in Fig.~\ref{fig:eta-v-etap} is explained by the fact that far away from $\Delta=\pi/2$, the approximate relation between the pseudospin correlation function and the $C(\vartheta)$ correlation function no longer holds.

\subsection{Scaling of the Cluster Size Distribution}
To further assess whether the \texttt{reflect} algorithm chooses the appropriate physical clusters, we measure the distribution of cluster sizes at $\Delta=0.5\pi$. Defining $n^*(s)$ as the probability density function of hitting on a cluster of size $s$, we obtain a bimodal distribution, with a large peak at small cluster sizes and a secondary peak at large cluster sizes. This indicates that the scaling of the average cluster size, from which the exponent $\eta'$ is determined, results from a subtle interplay between these two peaks. 

The cluster size distribution is typically studied in the context of the Swendsen-Wang algorithm, for which scaling forms have been derived, through the FK mapping to a critical percolation problem, to describe observables such as (1) the percolation probability $\langle P_\infty \rangle$, which exhibits the same finite size scaling $\sim L^{-\beta/\nu}$ as the net magnetization, (2) the distribution of the size of the largest (spanning) cluster in each bond configuration, whose mean scales as $L^{d_F}$ where $d_F$ is the fractal dimension, and (3) the number per site of clusters of size $s$~\cite{DOnorioDeMeo1990,Hou2018}. The scaling form of this third observable must be modified when considering a single-cluster algorithm, such as the Wolff and \texttt{reflect} algorithms. Under the hypothesis that the inclusion model at $\Delta=0.5\pi$ maps, via the above construction, to the critical point of some unknown percolation problem, we posit the scaling form~\cite{Leung1991}:
\begin{equation}\label{eq:scale}
    n^*(s, L) \sim s^{-\tau + 1}f(s/L^{d_F})
\end{equation}
where $\tilde{n}(x)$ is a universal scaling function.
\begin{figure}[t]
        \centering
        \includegraphics{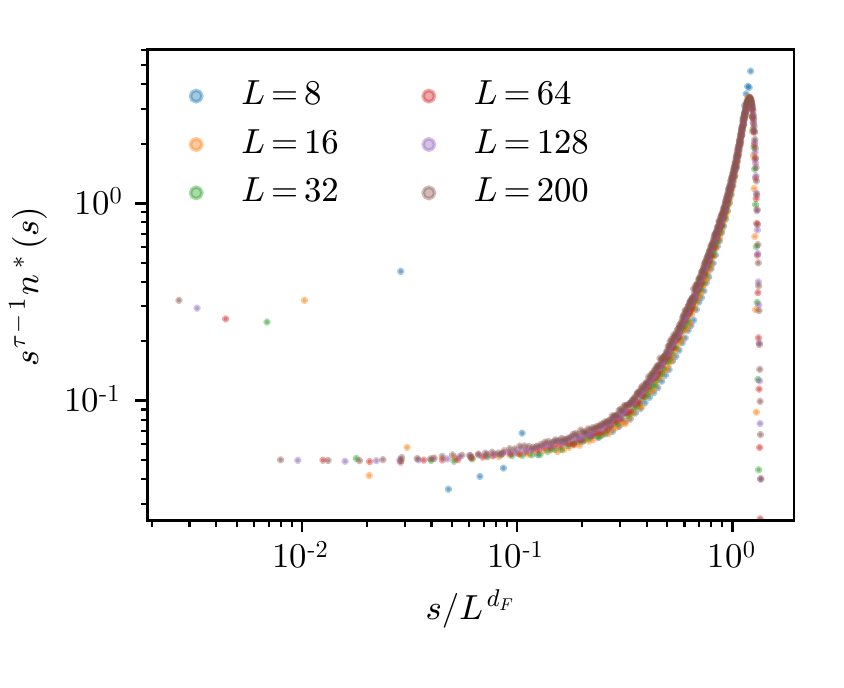}
        \caption{Rescaled probability distribution $\tilde{n}(s/L^{d_F})$ of cluster sizes at $\Delta=0.5\pi$. The horizontal axis is cluster size $s$ rescaled by $L^{d_F}$, where $d_F=1.901$ is the fractal dimension determined from Eq.~\ref{eq:dF}, using $\eta=0.198$. The vertical axis is $n^*(s)s^{\tau-1}$. Optimal data collapse was achieved using a cluster exponent of $\tau-1=1.089$, obtained from a power law fit described in the text. The distribution was measured over $\approx 8\times 10^5$ cluster moves for $L=8$, $\approx 1.1\times 10^6$ cluster moves for $L=16$, $\approx 1.65 \times 10^6$ cluster moves for $L=32$, $\approx 2.3 \times 10^6$ cluster moves for $L=64$, $\approx 3.2\times 10^6$ cluster moves for $L=128$, and $\approx 4\times 10^6$ moves for $L=200$.}
        \label{fig:distribution}
    \end{figure}    
 To verify this scaling form, in Fig.~\ref{fig:distribution} the cluster size distribution is multiplied by $s^{\tau - 1}$ and plotted as a function of $s/L^{d_F}$ for system sizes ranging from $L=8$ to $L=128$. The exponents $d_F$ and $\tau$ are related to the exponent $\eta'$ by summing Eq.~\ref{eq:scale} over $s$ to obtain the average cluster size:
\begin{align}\label{eq:dFT}
    L^{2-\eta'} \sim\langle s \rangle &= \sum_{s}  n^*(s) s\sim L^{d_F(3-\tau)} \notag \\
    \Rightarrow
    d_F &= \frac{2-\eta'}{3 - \tau}
\end{align}
Assuming hyperscaling, $d \nu = \gamma + 2\beta$~\cite{Goldenfeld1992}, the cluster exponent $\tau$ can be related to the fractal dimension $d_F$ via:
\begin{equation}\label{eq:tau1}
    \tau = 1 + d/d_F
\end{equation}
 which, when substituted into Eq.~\ref{eq:dFT} with $d=2$ implies:
 \begin{align}\label{eq:dF}
     d_F = 2 - \frac{\eta'}{2} = 2 - \frac{\eta}{2}.
\end{align}
Using the critical exponent $\eta=0.198$ at $\Delta=0.5\pi$, this implies a fractal dimension $d_F=1.901$. As a test of internal consistency, $d_F$ can be estimated via finite size scaling of  $s_{max}$, the cluster size at which $s^{\tau-1}n^*(s,L)$ is a maximum. A fit to the form $s_{max}\sim L^{d_F}$ yields $d_F=1.902 \pm 0.003$.

Substituting Eq.~\ref{eq:dF} into Eq.~\ref{eq:tau1} yields:
\begin{equation}\label{eq:tau}
\tau = 1 + \frac{4}{4-\eta}
\end{equation}
which, for $\eta=0.198$, implies a cluster exponent of $\tau = 2.052$. Again, we can verify the internal consistency of the scaling form by fitting $n^*(s) \sim s^{-\tau+1}$ in the regime $1\ll s \ll L^{d_F}$~\cite{Hou2018}. At the accessible system sizes, however, such a fit is prone to error, so we instead measure $\tau$ through a power-law fit to the value of $n^*(s,L)$ measured at the position of the secondary peak at system size $L$, vs. the position of the peak. This yields the estimate $\tau=2.089$, which is used to achieve the data collapse between different system sizes in Fig.~\ref{fig:distribution}. The good data collapse at large cluster sizes, particularly for $L\geq 16$, provides evidence in favor of the assumed scaling form as the leading behavior of the distribution. 
 
 Taken together, Figs.~\ref{fig:eta-v-etap} and~\ref{fig:distribution} illustrate the connection between the scaling exponents of the original spin model and the cluster model derived from the \texttt{reflect} algorithm. Combined with the ergodicity proven in Appendix~\ref{sect:app1}, this suggests that the algorithm is, in the same sense as the original Swendsen-Wang/Wolff algorithm for the $d=2$ Ising model, optimally adapted to our model at the value $\Delta = 0.5 \pi$. 
 
Our value of $\tau$ is quite close to the exact value $187/91$ for short-ranged percolation in $d=2$, which leads us to conjecture that our algorithm is constructing clusters precisely at the critical point of that problem. Further, if we use the exact value of $\tau$ for the scaling collapse, the results are equally impressive by eye. 

\subsection{Long-ranged Bond Correlations}
We note that the connection to short-ranged percolation is somewhat unexpected as the underlying spin correlations are long-ranged. Decomposing the entire lattice into clusters \'a la Swendsen-Wang helps clarify this picture. The first moment of the bond distribution is measured to be $\langle b \rangle = 1/2$, in agreement with the percolation threshold $p_c$ of uncorrelated bond percolation on the square lattice on the grounds of duality~\cite{Christensen2002}. However, in contrast to that short-ranged percolation problem, we do find evidence of QLRO in the bond correlations. 
 Figure~\ref{fig:bond-corr} shows the finite size scaling of the bond susceptibility, defined analogously to the spin susceptibility as:
 \begin{equation}\label{eq:chi-b}
    \chi_b(L) = \frac{1}{L^2} \sum_{\langle ij \rangle\atop \langle kl \rangle} \left[\langle b_{ij} b_{kl} \rangle - \langle b \rangle^2\right].
\end{equation}
We find that this scales as a power law $\chi_b\sim L^{2-a}$, implying algebraic decay of the bond correlations with exponent $a$:
\begin{equation}\label{eq:lr-corr}
    \langle b_{ij} b_{kl} \rangle - \langle b \rangle^2 \propto 1/r^a
\end{equation}
 where $r$ is the distance between bonds $\langle ij \rangle$ and $\langle kl \rangle$ on the lattice. The estimated value of $a$, determined from scaling up to $L=200$, is:
 \begin{equation}\label{eq:a}
     a = 0.7870 \pm 0.00066.
 \end{equation}
\begin{figure}[t]
     \centering
     \includegraphics[width=\linewidth]{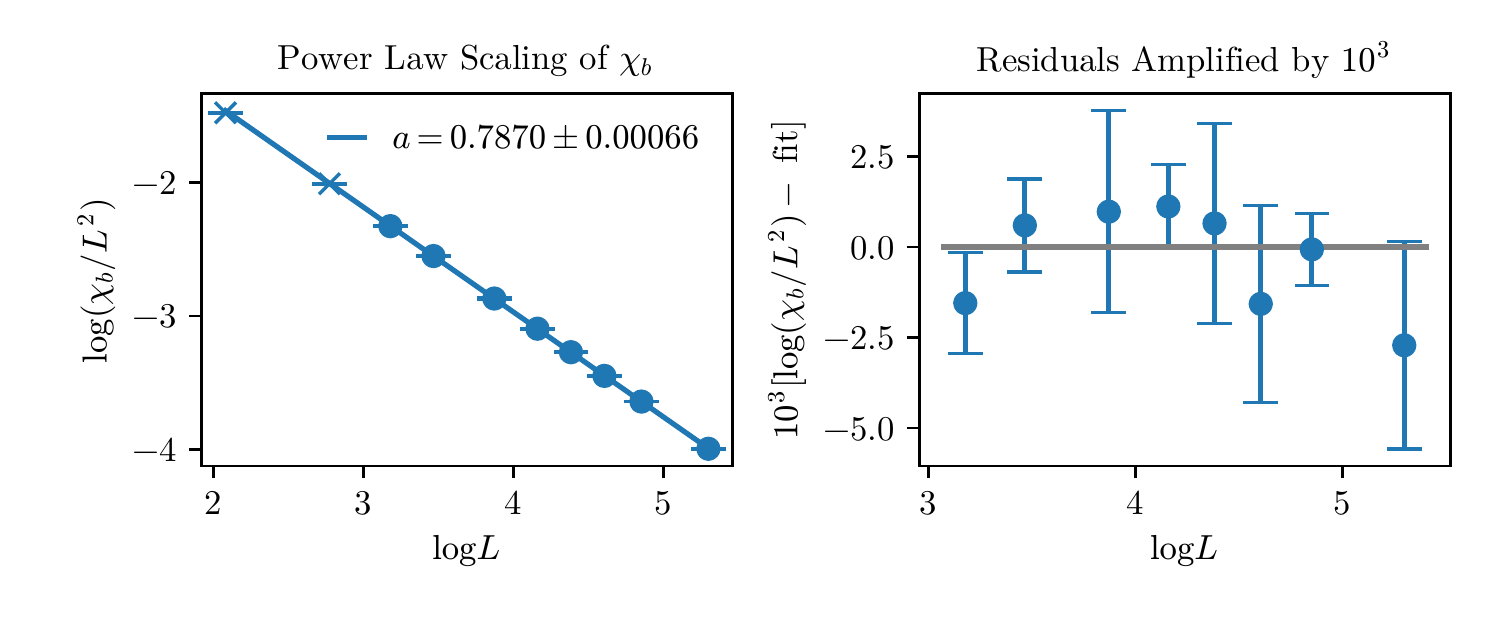}
     \caption{Finite size scaling of the bond susceptibility, defined in Eq.~\ref{eq:chi-b}, for $L=8,16,24,32,48,64,80,100,128,200$. Only data with $L\geq 24$ was included in the fit. Error bars were determined from the standard error across 100 or 1000 chunks. Right panel shows the residuals with respect to the fit, amplified by a factor of $10^3$.}
     \label{fig:bond-corr}
 \end{figure} 
 The Harris criterion applied to long-range correlated percolation suggests that correlations of the form of Eq.~\ref{eq:lr-corr} are irrelevant for $a>2/\nu$, where $\nu$ is the percolation correlation length exponent for the short-ranged percolation problem~\cite{Weinrib1984}. In this regime, the critical exponents at the percolation threshold are unaffected by the presence of long-range bond correlations. When the correlations become relevant ($a<2/\nu$), the critical exponents are expected to deviate from the uncorrelated percolation exponents in an $a$-dependent fashion.

Given that $\nu=4/3$ for short-ranged 2D percolation, Eq.~\ref{eq:a} implies that the quasi-long-ranged bond correlations are in fact relevant. However, simulation of long-range correlated site percolation on the square lattice~\cite{Prakash1992} indicates that while other critical exponents depend on $a$, the fractal dimension $d_F$ (and in turn, $\tau$) remains consistent with its uncorrelated value as $a$ decreases. Similarly, on the triangular lattice $d_F$ was found to be independent of $a$ down to $a=2/3$, below which $d_F$ increases continuously~\cite{Schrenk2013}. Thus, we conclude that at $\Delta=0.5\pi$, the cluster model produced by the \texttt{reflect} algorithm maps to standard percolation in $d=2$, with relevant quasi-long-ranged correlations that do not affect the scaling of average or largest cluster size. To further investigate this interpretation, in future research it will be worthwhile to measure other critical exponents of the percolation model that are sensitive to these correlations.
\end{document}